\newif\ifstoc
\stoctrue %
\stocfalse

\documentclass[11pt,fleqn,letterpaper]{article}
\usepackage{fullpage, amsthm}

\usepackage{amsmath}
\usepackage{amssymb}
\usepackage{xspace}
\usepackage{color}
\usepackage{graphicx}

\usepackage{ifpdf}
\ifpdf    %
\usepackage{hyperref}
\else    %
\usepackage[hypertex]{hyperref}
\fi

\newtheorem{theorem}{Theorem}[section]
\newtheorem{lemma}[theorem]{Lemma}
\newtheorem{prop}[theorem]{Proposition}
\newtheorem{claim}[theorem]{Claim}

\newtheorem{corollary}[theorem]{Corollary}

\theoremstyle{definition}

\newcommand {\ignore} [1] {}
\newcommand{\vzero}{\underline 0}

\newcommand{\veps}{\varepsilon}
\newcommand{\R}{{\mathbb R}}
\newcommand{\D}{{\mathcal D}}
\newcommand{\zo}{\{0,1\}}

\DeclareMathOperator{\conv}{conv}
\DeclareMathOperator*{\EX}{{\mathbb E}}

\newcommand{\minn}[1]{\min\{{#1}\}}
\newcommand{\maxx}[1]{\max\{{#1}\}}

\providecommand{\card}[1]{\lvert#1\rvert}

\providecommand{\aset}[1]{\{#1\}}
\providecommand{\eqdef}{:=}
\providecommand{\discrete}{\text{discrete}}
\providecommand{\Ddisc}{\D_\veps^\discrete}
\providecommand{\hatG}{{\smash{\widehat{G}}}}
\providecommand{\tilp}{\smash{\widetilde{p}}}
\providecommand{\whp}{with high probability\xspace}

\def\compactify{\itemsep=0pt \topsep=0pt \partopsep=0pt \parsep=0pt}

\newcounter{this-list}

\newcounter{par-list}
\newlength{\parlistlength}

\newcommand{\dem}{\mathbf{d}}
\newcommand{\len}{\ell}
\newcommand{\dist}{\delta}
\newcommand{\maxeps}{{\tfrac18}}
\providecommand{\hatf}{{\smash{\widehat{f}}}}

\renewcommand{\epsilon}{\varepsilon}
\newcommand{\eps}{\varepsilon}

\newcommand{\gap}{\gamma}
\newcommand{\T}{\mathcal{T}}

\newcommand{\mincut}{\operatorname{mincut}}
\newcommand{\poly}{\operatorname{poly}}
\newcommand{\sse}{\subseteq}

\newcommand{\etal}{et al.\xspace}

\newenvironment{myproof}{\begin{proof}}{\end{proof}}

\allowdisplaybreaks

\title{Towards $(1+\epsilon)$-Approximate Flow Sparsifiers%
}

\author{
Alexandr Andoni\thanks{%
Email: \texttt{andoni@microsoft.com} }
\\ Microsoft Research
\and
Anupam Gupta\thanks{%
Work supported in part by
    NSF awards CCF-0964474 and CCF-1016799, US-Israel BSF grant
    \#2010426, and by a grant from the
    CMU-Microsoft Center for Computational Thinking. Email: \texttt{anupamg@cs.cmu.edu}}
\\ CMU and MSR
\and
Robert Krauthgamer\thanks{%
Work supported in part by the Israel Science Foundation grant \#897/13, 
the US-Israel BSF grant \#2010418, and by the Citi Foundation.
Email: \texttt{robert.krauthgamer@weizmann.ac.il} }
\\ Weizmann Institute
}

\begin{document}
\maketitle

\begin{abstract}
  A useful approach to ``compress'' a large network $G$ is to
  represent it with a {\em flow-sparsifier}, i.e., a small network $H$
  that supports the same flows as $G$, up to a factor $q \geq 1$
  called the quality of sparsifier.  Specifically, we assume the
  network $G$ contains a set of $k$ terminals $T$, shared with the
  network $H$, i.e., $T\subseteq V(G)\cap V(H)$, and we want $H$ to
  preserve all multicommodity flows that can be routed between the terminals
  $T$. The challenge is to construct $H$ that is small.

  These questions have received a lot of attention in recent years,
  leading to some known tradeoffs between the sparsifier's quality $q$ and
  its size $|V(H)|$. Nevertheless, it remains an
  outstanding question whether every $G$ admits a
  flow-sparsifier $H$ with quality $q=1+\epsilon$, or even $q=O(1)$,
  and size $|V(H)|\leq f(k,\epsilon)$ (in particular, independent of 
  $|V(G)|$ and the edge capacities).

  Making a first step in this direction, we present new constructions 
  for several scenarios:
\begin{itemize}
\item
Our main result is that for quasi-bipartite networks $G$, one can construct
a $(1+\epsilon)$-flow-sparsifier of size $\poly(k/\eps)$. 
In contrast, exact ($q=1$) sparsifiers for this family of networks
are known to require size $2^{\Omega(k)}$.
\item
For networks $G$ of bounded treewidth $w$, we construct a flow-sparsifier 
with quality $q=O(\log w / \log\log w)$ and size $O(w\cdot \poly(k))$.
\item
For general networks $G$, we construct a {\em sketch} $sk(G)$,
that stores all the feasible multicommodity flows up to factor $q=1+\eps$,
and its size (storage requirement) is $f(k,\eps)$.
\end{itemize}

\end{abstract}

\newpage

\section{Introduction}
\label{sec:intro}

A powerful tool to deal with big graphs is to ``compress'' them by
reducing their size --- not only does it reduce their storage requirement,
but often it also reveals opportunities for more efficient graph algorithms. 
Notable examples in this context include the cut and spectral sparsifiers 
of \cite{BK96, ST04a}, which have had a huge impact on graph algorithmics. 
These sparsifiers reduce the number of edges of the graph, while preserving
prominent features such as cut values and Laplacian spectrum, up to 
approximation factor $1+\eps$. This immediately improves the runtime of graph
algorithms that depend on the number of edges, at the expense of 
$(1+\eps)$-approximate solutions. 
Such sparsifiers reduce only the number of {\em edges},
but it is natural to wonder whether more is to be gained by reducing
the number of {\em nodes} as well. This vision --- of ``node sparsification'' 
--- already appears, say, in \cite{FM95}. 

One promising notion of node sparsification is that of {\em flow or cut
  sparsifiers}, introduced in \cite{HKNR98, Moitra09, LM10}, where we have a
\emph{network} (a term we use to denote edge-capacitated graphs) $G$,
and the goal
is to construct a smaller network $H$ that supports the same flows as
$G$, up to a factor $q \geq 1$ called the quality of sparsifier $H$.
Specifically, we assume the network $G$ contains a set $T$ of $k$
\emph{terminals} shared with the network $H$, i.e., $T\subseteq V(G)\cap
V(H)$, and we want $H$ to preserve all multicommodity flows that can
be routed between the terminals $T$. 
(A formal definition is given in Section~\ref{sec:prelims}.)
A somewhat simpler variant is a {\em cut sparsifier}, 
which preserves the single-commodity flow from every set $S\subset T$ to its complement
$T\setminus S$, i.e., a minimum-cut in $G$ of the terminals
bipartition $T=S\cup (T\setminus S)$.
Throughout, we consider undirected networks (although some of the results
apply also for directed networks), 
and unless we say explicitly otherwise, flow and cut sparsifiers refer
to their node versions, i.e.,
networks on few nodes that support (almost) the same flow. 

The main question is: \emph{what tradeoff can one achieve between the
  quality of a sparsifier and its size?}  This question has received a
lot of attention in recent years.  In particular, if the sparsifier is
only supported on $T$ (achieves minimal size), one can guarantee quality
$q \leq O\big(\frac{\log k}{\log\log k}\big)$ \cite{Moitra09, LM10,
  CLLM10, EGKRTT10, MM10}.  On the other hand, with this minimal size,
the (worst-case) quality must be $q\geq \tilde \Omega(\sqrt{\log k})$
\cite{LM10, CLLM10, MM10}, and thus a significantly better quality
cannot be achieved without increasing the size of the sparsifier.  The
only other result for flow sparsifiers, due to \cite{Chuzhoy12},
achieves a constant quality sparsifiers whose size depends on the
capacities in the original graph. (Her results give %
flow sparsifiers of size $C^{O(\log\log C)}$; here $C$ is the total
capacity of edges incident to terminals and hence may be $\Omega(nk)$
even for unit-capacity graphs.) For the simpler notion of cut
sparsifiers, there are known constructions at the other end of the
tradeoff. Specifically, one can achieve exact (quality $q=1$) cut sparsifier of
size $2^{2^{k}}$ \cite{HKNR98, KRTV12}, however, the size must still be
at least $2^{\Omega(k)}$ \cite{KRTV12, KR13} (for both cut and flow
sparsifiers).

Taking cue from edge-sparsification results, and the above lower bounds,
it is natural to focus on small sparsifiers that achieve quality
$1+\eps$, for small $\eps \geq 0$. Note that for flow sparsifiers, we do
not know of any bound on the size of the sparsifier that would depend
only on $k$ (and $1/\eps$), but not on $n$ or edge capacities. In fact,
we do not even know whether it is possible to represent the sparsifier
{\em information theoretically} (i.e., by a small-size sketch), let alone by a
graph.

\subsection{Results} 
\label{sec:results}

Making a first step towards constructing high-quality sparsifiers of
small size, we present constructions for several scenarios:

\begin{itemize} \compactify
\item
Our main result is for \emph{quasi-bipartite graphs}, i.e., graphs where the
non-terminals form an independent set (see~\cite{RV99}),
and we construct for such networks 
a $(1+\epsilon)$-flow-sparsifier of size $\poly(k/\eps)$. 
In contrast, exact ($q=1$) sparsifiers for this family of networks
are known to require size $2^{\Omega(k)}$ \cite{KRTV12, KR13}. (See Theorem \ref{thm:bipartiteNew}.)
\item
For general networks $G$, we construct a {\em sketch} $sk(G)$, that
stores all the feasible multicommodity flows up to factor $q=1+\eps$,
and has size (storage requirement) of $f(k,\eps)$ words. This implies an
affirmative answer to the above information-theoretic question on
existence of flow sparsifiers, and raises the hope for a
$(1+\eps)$-flow-sparsifier of size $f(k,\eps)$. (See Theorem \ref{thm:DS}.)
\item
For networks $G$ of bounded treewidth $w$, we construct a flow-sparsifier 
with quality $q=O(\frac{\log w}{\log\log w})$ and size $O(w\cdot \poly(k))$. (See Theorem \ref{thm:twk}.)
\item 
Series-parallel networks admit an exact (quality $1$) 
flow sparsifier with $O(k)$ vertices. (See Theorem \ref{thm:sp}.)
\end{itemize}

\subsection{Techniques}
\label{sec:techniques}

Perhaps our most important contribution is the introduction of the 
three techniques listed below, and indeed, one can view our results 
from the prism of these three rather different approaches.
In particular, applying these three techniques to quasi-bipartite graphs
yields $(1+\eps)$-quality sparsifiers whose sizes are (respectively) 
doubly-exponential, exponential, and polynomial in $k/\eps$.

\begin{enumerate} \compactify
\item \emph{Clumping:} %
  We first ``discretize'' the set of (almost) all possible
  multi-commodity demands into a finite set, whose size depends only on
  $k/\eps$, and then partition the graph vertices into a small number
  of ``clusters'', so that clumping each cluster into a single vertex 
  still preserves one (and eventually all) of the discretized demands. 
  The idea of clumping vertices was used in the past
  to obtain exact (quality $1$) cut sparsifiers \cite{HKNR98}.
  Flow-sparsifiers require, in effect, to preserve all metrics between 
  the terminals rather than merely all inter-terminal cut metrics,
  and requires new ideas.

\item \emph{Splicing/Composition:} Our Splicing Lemma shows that it is
  enough for the sparsifier to maintain flows routed using paths that do
  not contain internally any terminals. Our Composition Lemma shows that for a
  network obtained by gluing two networks along some subset of their
  terminals, gluing the respective sparsifiers (in the same manner) gives us
  a sparsifier for the glued network. These lemmas enable us to do
  ``surgery'' on networks, to decompose and recompose them, so that we
  find good sparsifiers on smaller networks and then combine them
  together without loss of quality.
\item \emph{Sampling:} 
  This technique samples parts of the graph, while preserving the flows 
  approximately. The main difficulty is to determine correct sampling 
  probabilities (and correlations). 
  This is the technical heart of the paper, and we outline its main ideas 
  in Section~\ref{sec:intro-sample}.
\end{enumerate}

We hope they will inspire ulterior constructions of high-quality 
flow sparsifiers for general graphs.
The clumping techniques
give information-theoretic bounds on flow-sparsification, and the
splicing/composition approach proves useful for sparsification of bounded
treewidth and series-parallel graphs (beyond what can be derived 
using their flow/cut gaps from known cut sparsifiers).

\subsection{Outline of Our Sampling Approach}
\label{sec:intro-sample}

A classic approach to obtain an {\em edge}-sparsifier \cite{Karger94,BK96,SS11}
is to sample the edges of the graph and rescale appropriately. 
Here, we outline instead how to sample the \emph{vertices} of the graph 
to obtain a small flow-sparsifier. 
We outline our main idea on quasi-bipartite graphs (where the
non-terminals form an independent set), considering 
for simplicity the (simpler) question of {\em cut} sparsifiers, 
where we want to construct a smaller graph $G'$ that preserves the minimum 
cut between every bipartition of terminals $T=S\cup (T\setminus S)$.
The main idea is to sample a small number of non-terminals $v$,
keeping only their incident edges, and rescaling the corresponding capacities.
For
a fixed bipartition $T=S\cup \bar S$, we can write the value of the
min-cut as
\begin{equation}
\label{eqn:cutSum}
\alpha_{S,\bar S}=\sum_{v\notin T} \min\Big\{\sum_{s\in S}c_{s v},\sum_{t\in \bar S}c_{v t}\Big\}.
\end{equation}
(Here $c_{xy}$ is the capacity of the edge $xy$.)
Suppose we assign each non-terminal $v$ with some sampling probability $p_v$, 
then sample the non-terminals using these probabilities, letting
$I_v$ be an indicator variable for whether $v$ was sampled.  Then, for
sampled $v$'s we re-normalize the capacities on incident edges by
$1/p_v$, i.e., the new capacities are $c'_{v,t}=c_{v,t}/p_v$ for all
$t\in T$ (non-sampled $v$'s are dropped). The new value of the min-cut
in the sparsifier $G'$ is
\begin{equation}
\label{eqn:cutApprox}
\alpha'_{S,\bar S}=\sum_{v\notin T} I_v/p_v\cdot \min\Big\{\sum_{s\in S}c_{s 
v},\sum_{t\in \bar S}c_{v t}\Big\}.
\end{equation}
This classical estimator is unbiased, and hence each min-cut
$\alpha_{S,\bar S}$ is preserved {\em in expectation}.

The main challenge now is to prove that the above random sum
concentrates around its expectation for ``small'' values of
$p_v$. For example, consider setting all $p_v$ equal, say to 
$\poly(k)/|V|$. Even if all $c_{s,v}\in\zo$ (i.e., all edges in $E$ 
have unit capacity), due to the min operation, it is possible 
that only very few terms in the summation in Eqn.~\eqref{eqn:cutSum} 
have nonzero contribution to $\alpha_{S,\bar S}$, and are 
extremely unlikely to be sampled.

Our general approach is to employ {\em importance sampling}, where
each $p_v$ is related to $v$'s contribution to the sum, namely
$\min\left\{\sum_{s\in S}c_{s v},\sum_{t\in \bar S}c_{v t}\right\}$. 
Applying this directly is hard --- since that minimum 
depends on the bipartition $S\cup \bar S$, whereas $p_v$
cannot. Instead, we exploit the fact that for any bipartition,
we can estimate
\begin{equation}\label{eqn:cutST}
  \alpha'_{S,\bar S} 
  \geq \max_{s\in S, t\in \bar S}\ \sum_v I_v/p_v\cdot \min\{c_{s v},c_{v t}\}
  \geq \tfrac{1}{k^2} \alpha'_{S,\bar S},
\end{equation}
and hence arguing about the sum
in Eqn.~\eqref{eqn:cutST} should be enough for bounding
the variance. Following this reasoning through, it turns out that a good choice is
\begin{equation}\label{eqn:cutPv}
p_v=M\cdot \max_{s\neq t}
\frac{\min\{c_{s v},c_{v t}\}}{\sum_{v'} \min\{c_{s v'},c_{v' t}\}},
\end{equation}
where $M=\poly(k/\eps)$ is an over-sampling factor. The underlying
intuition of Eqn.~\eqref{eqn:cutPv} is that, replacing the max with 
a ``correct'' choice of $s\in S,t\in \bar S$, the denominator is just the entire
potential contribution 
to the sum in Eqn.~\eqref{eqn:cutST}, and hence these $p_v$ values
can be used as importance sampling probabilities for the sum in
Eqn.~\eqref{eqn:cutApprox}.  Moreover, we prove that this setting of $p_v$
allows for a high-probability concentration bound in the sum from
Eqn.~\eqref{eqn:cutApprox}, and thus sampling $\poly(k/\eps)$ vertices 
suffices for the purpose of taking a union bound over all $2^k$ bipartitions.

So far we have described the approach for obtaining cut sparsifiers,
but in fact we prove that the exact same approach works for obtaining
flow sparsifiers as well. There are more issues that we need to take
care in this generalized setting. First, we need to bound the
``effective'' number of demand vectors. Second, the flow does not have 
a simple closed-form formula like \eqref{eqn:cutSum}, so upper
and lower bounds need to be proved by analyzing separately (the concentration
of) the optimal value of the flow LP and of its dual.

\section{Preliminaries}
\label{sec:prelims}

A \emph{$k$-terminal network} is an edge-capacitated graph $G=(V,E,c)$
with a subset $T\subseteq V$ of $k$ terminals. We will be interested
only in \emph{terminal flows}, i.e., flows that start and end only at
the terminal vertices of $G$. Define $\D(G)$, the \emph{demand polytope
  of $G$} as the set of all demand vectors $\dem$ that are supported only
on terminal-pairs, and admit a feasible multicommodity-flow in $G$,
formally,
\begin{gather}
  \D(G) \eqdef \aset{ \dem \in \R_+^{\binom{T}{2}} :\ \text{demand
      $\dem$ can be routed in $G$} }, \label{eq:1}
\end{gather}
where we denote $\R_+\eqdef \aset{x\in R: x\ge 0}$ and
$\tbinom{T}{2}\eqdef \aset{S\subseteq T:\ \card{S}=2 }$.  Throughout,
we assume $G$ is connected.

\begin{lemma}
  \label{lem:polytope}
  $\D(G)$ is a polytope, and is down-monotone.
\end{lemma}

\begin{myproof}
  Let $\mathcal{P}_{ij}$ be the set of paths between terminals $i$ and
  $j$. Consider the extended demand polytope $\D_{ext}(G)$ with
  variables $d_{ij}$ for all $\{i,j\} \in \binom{T}{2}$, and $f_P$ for
  each $P \in \cup_{i,j} \mathcal{P}_{ij}$. 
  \begin{align*}
    \textstyle \sum_{P \in \mathcal{P}_{ij}} f_P &= d_{ij} \\
    \textstyle \sum_{i,j} \sum_{P \in \mathcal{P}_{ij}: e \in P} f_P &\leq c_e \\
    f_P, d_{ij} &\geq 0.
  \end{align*}
  This polytope captures all the feasible terminal flows, and hence all
  the routable demands between the terminals of $G$. The projection of
  this polytope $\D_{ext}(G)$ onto the variables $d$ is exactly $\D(G)$;
  hence the latter is also a polytope.\footnote{As a aside, we can write
  $\D_{ext}(G)$ more compactly using edge-flow variables $f^{ij}(e)$
  instead of path variables $f_P$; we omit such standard optimizations
  here.} Finally, the down-monotonicity of the polytope follows from the
  downward-feasibility of flows, in turn due to the lack of lower-bounds
  on the flows on edges.
\end{myproof}

\paragraph{Dual linear program for concurrent flow.}

For a demand vector $\dem \in \R_+^{\binom{T}{2}}\setminus\aset{0}$, 
we denote the
\emph{concurrent flow} problem (inverse of the congestion) by
\[
  \lambda_G(\dem) \eqdef \sup \aset{ \lambda\ge0 :\ \lambda \dem \in \D(G) }.
\]
This is well-defined because $\vec0\in \D(G)$. The following well-known
lemma writes $\lambda_G(\dem)$ by applying linear programming (LP) duality to
multicommodity flow, see e.g.\ \cite{LR99,Shmoys:CutSurvey,Moitra09}.

\begin{lemma} \label{lem:LP} $\lambda_G(\dem)$ can be computed via the
  linear program \eqref{eq:dualLP1} %
  which has ``edge-length''
  variables $\len_e$ for edges $e\in E$ and ``distance'' variable
  $\dist_{uv}=\dist_{vu}$ for terminal pairs
  $\aset{s,t}\in\binom{T}{2}$.

\begin{equation} \label{eq:dualLP1}
\framebox{ $
\begin{array}{lllll}
  \lambda_G(\dem) = 
  & \min 
  &  \sum_{e\in E} c_e \len_e 
  \\
  & \operatorname{s.t.}
  &  \sum_{\aset{s,t}\in\tbinom{T}{2}} d_{st} \dist_{st} \ge 1 
  \\
  & &  \dist_{st} \le \sum_{e\in P} \len_e 
  & \forall \aset{s,t}\in\tbinom{T}{2} \text{ and a path $P$ connecting them}
  \\
  & & \len_e \ge 0 
  & \forall e\in E
  \\
  & & \dist_{st} \ge 0 
  & \forall \aset{s,t}\in\tbinom{T}{2}.
\end{array}
$ } \tag{LP1}
\end{equation}

\end{lemma}

\paragraph{Flow-sparsifier definition.}

A network
$G'=(V',E',c')$ with $T \subseteq V'$ is called a \emph{flow sparsifier
of $G$ with quality $q\ge 1$} if
\[
  \forall d\in \R_+^{\binom{T}{2}},
  \qquad
  \lambda_G(\dem) \leq \lambda_{G'}(\dem) \leq q\cdot \lambda_G(\dem).
\]
This condition is equivalent to writing $\D(G) \subseteq \D(G') \subseteq q\cdot\D(G)$.

\section{A Data Structure for Multicommodity Flows} \label{sec:DS}

We present a data structure that ``maintains'' $\D(G)$ within
approximation factor $1+\veps$.  More precisely, we preprocess the
terminal network $G$ into a data structure whose storage requirement
depends only on $k$ and $\veps$ (but not on $n=|V(G)|$).  Given a
query $\dem \in \R_+^{\binom{T}{2}}$, this data structure returns an
approximation to $\lambda_G(\dem)$ within factor $1+\veps$ (without
further access to $G$).  The formal statement appears in Theorem
\ref{thm:DS}.  We assume henceforth that $0<\veps<\maxeps$.

\paragraph{An approximate polytope.}

Let $G=(V,E,c)$ be a terminal network with k terminals $T\subset V$. 
For each commodity $\aset{i,j} \in\binom{T}{2}$, let $L_{ij}$ be the maximum flow of
commodity $ij$ alone (i.e., as a single-commodity flow) in $G$.
Discretize the set $\D(G)$ defined in~\eqref{eq:1} by defining the subset
\[
  \Ddisc \eqdef \aset{\, \dem \in \D(G) :\ \text{every nonzero $d_{ij}$ 
is a power of $1+\veps$ in the range $[\veps/k^2\cdot L_{ij},L_{ij}]$\, }  }.
\]
The range upper bound $L_{ij}$ (which is not really necessary, 
as it follows from $\dem\in \D(G)$), immediately implies that 
\begin{equation} \label{eq:DdiscSize}
  \card{\Ddisc} 
  \leq \Big(1+\tfrac1\veps \log_{1+\veps} k\Big)^{\binom{k}{2}}
  \leq \Big(\tfrac{O(1)}\veps \log k\Big)^{k^2}.
\end{equation}

\begin{lemma} \label{lem:monotone}
The convex hull $\conv(\Ddisc)$ is down-monotone,
namely, if $\dem \in\conv(\Ddisc)$ and $0\leq \widehat\dem \leq \dem$, 
then also $\widehat \dem \in\conv(\Ddisc)$.
\end{lemma}

\begin{myproof}
Consider first the special case where $\widehat \dem$ is obtained from $\dem$ 
by scaling the coordinates in some subset $S\subseteq \binom{T}{2}$
by a scalar $0\le \beta <1$.
Write $\dem$ as a convex combination of some vectors $\dem_j\in \Ddisc$,
say $\sum_j \alpha_j \dem_j$, where $\alpha_j>0$ and $\sum_j \alpha_j = 1$.
Let $\widehat \dem_j$ be the vector obtained from $\dem_j$ by zeroing 
all the coordinates in $S$, and observe it is also in $\Ddisc$.
Now write 
\[
  \widehat \dem 
  = \sum_{j} \alpha_j [\beta\dem_j + (1-\beta)\widehat\dem_j]
  = \sum_{j} \alpha_j\beta \dem_j + \sum_{j} \alpha_j(1-\beta) \widehat \dem_j,
\]
and observe the right-hand side is a convex combination of vectors in $\Ddisc$, 
which proves the aforementioned special case.
The general case follows by iterating this special-case argument several times.
\end{myproof}

\paragraph{The data structure.}

The next theorem expresses the space requirement of an algorithm
in machine words, assuming every machine word can store $\log k$ bits 
and any single value $L_{ij}$ 
(either exactly or within accuracy factor $1+\veps/2$).
This holds, in particular, when edge capacities in the graph $G$ 
are integers bounded by $n=\card{V(G)}$, and a word has $2\log n$ bits.

\begin{theorem} \label{thm:DS}
For every $0<\veps<\maxeps$ there is a data structure 
that provides a $(1+\veps)$-approximation for multicommodity flows
in a $k$-terminal network, 
using space $\Big(\tfrac{O(1)}{\veps} \log k\Big)^{k^2}$ 
and query time $O(\frac1\veps k^2 \log k)$.
\end{theorem}

\begin{myproof}
We present a data structure achieving approximation $1+O(\veps)$;
the theorem would then follow by scaling $\veps>0$ appropriately.
The data structure stores the set $\Ddisc$, 
using a dictionary (such as a hash table) to answer membership queries 
in time $O(k^2)$, the time required to read a single vector. 
It additionally stores all the values $L_{ij}$.
(We assume these values can be stored exactly; the case of $1+\veps/2$ 
approximation follows by straightforward modifications.)

Given a query $\dem$, the algorithm first computes 
$\beta = \min_{i,j\in T}\aset{L_{ij}/d_{ij}}$.
We thus have that $\beta k^{-2} \leq \lambda_G(\dem) \leq \beta$,
because the commodity $i,j\in T$ attaining $\beta d_{ij}=L_{ij}$
limits $\lambda_G(\dem)$ to not exceed $\beta$,
and because we can ship $\beta d_{ij}\leq L_{ij}$ units separately 
for every commodity $i,j\in T$, hence also their convex combination 
$\binom{k}{2}^{-1}\beta \dem$.

The query algorithm then computes an estimate for $\lambda_G(\dem)$ 
by performing a binary search over all powers of $(1+\veps)$ 
in the range $[\beta k^{-2},\beta]$,
where each iteration decides, up to $1+2\veps$ multiplicative approximation,
whether a given $\lambda$ in that range is at most $\lambda_G(\dem)$.
The number of iterations is clearly 
$O(\log_{1+\veps} k^2) \le O(\frac1\veps \log k)$.

The approximate decision procedure is performed in two steps.
In the first step, we let $\dem^-$ be the vector obtained from $\lambda \dem$
by zeroing all coordinates that are at most $2\veps/k^2\cdot L_{ij}$.
This vector can be written as 
\[ 
  \dem^-=\lambda \dem-\sum_{i,j\in T} \alpha_{ij} L_{ij} \mathbf{e}_{ij},
\]
where $\mathbf{e}_{ij}\in\R_+^{\binom{T}{2}}$ is the standard basis
vector corresponding to $\aset{i,j}$, and for every $i,j\in T$ we
define $\alpha_{ij} \eqdef d_{ij}/L_{ij}$ 
if $d_{ij} \leq 2\veps/k^2\cdot L_{ij}$, and $\alpha_{ij}\eqdef 0$ otherwise.
By definition, $\sum_{i,j} \alpha_{i,j}\leq \veps$.
The second step lets $\dem'$ be the vector obtained from $\dem^-$ 
by rounding down each coordinate to the nearest power of $1+\veps$.
Finally, decide whether $\lambda \leq \lambda_G(\dem)$ 
by checking whether $\dem'\in \Ddisc$, 
which is implemented using the dictionary in $O(k^2)$ time.

It remains to prove the correctness of the approximate decision procedure.
For one direction, assume that $\lambda \leq \lambda_G(\dem)$.
It follows that the demands $\lambda \dem \ge \dem^- \ge \dem'$ 
can all be routed in $G$, and furthermore $\dem' \in \Ddisc$,
implying that our procedure reports a correct decision.
For the other direction, 
suppose our procedure reports that $\lambda \leq \lambda_G(\dem)$,
which means that its corresponding $\dem'\in \Ddisc \subset \D(G)$.
We can thus write
\[ 
  \lambda \dem 
  = \dem^- + \sum_{ij} \alpha_{ij} L_{ij} \mathbf{e}_{ij} 
  \leq (1+\veps)\dem' + \sum_{ij} \alpha_{ij} L_{ij} \mathbf{e}_{ij}.
\]
The right-hand side can be described as a positive combination of vectors 
in $\D(G)$, whose sum of coefficients is 
$(1+\veps)+\sum_{ij}\alpha_{ij} \le 1+2\veps$.
Since $\D(G)$ is convex and contains $\vzero$,
we have that also $(1+2\veps)^{-1}\lambda\dem \in \D(G)$,
i.e., that $(1+2\veps)^{-1}\lambda \le \lambda_G(\dem)$,
which proves the correctness of the decision procedure 
up to $1+2\veps$ multiplicative approximation.
Overall, we have indeed shown that the binary search algorithm approximates 
$\lambda_G(\dem)$ within factor $(1+\veps)(1+2\veps) \leq 1+O(\veps)$.
\end{myproof}

\section{The Clumping Method for Flow Sparsifiers} 
\label{sec:clump}

In this section we develop a method based on clumping (merging) vertices, 
and exemplify its use on quasi-bipartite graphs.
Let $G=(V,E,c)$ be a terminal network with terminal set $T$. 
For a subset $S\subseteq V$, denote the edges in the induced subgraph
$G[S]$ by $E[S] \eqdef \aset{(u,v)\in E:\ u,v\in S}$. Given a partition
$\Pi = \{ S_1, S_2, \ldots, S_{m}\}$ of the vertex set (i.e., $\cup_{i =
  1}^m S_i = V$), say a distance function $\dist$ on the vertex set $V$
is \emph{$\Pi$-respecting} if for all $l \in [m]$ and $\{i,j\} \in S_l$
it holds that $\dist_{ij} = 0$.

\begin{prop} \label{prop:SuffCond2} Let $G=(V,E,c)$ be a $k$-terminal
  network, and fix $0<\veps<1/3$ and $b\ge1$.  Suppose there is an
  $m$-way partition $\Pi = \{S_1, \ldots, S_{m}\}$ such that for every
  $\dem'\in\Ddisc$, there exists a $\Pi$-respecting distance function
  $\dist$ that is a feasible solution to~(\ref{eq:dualLP1}) with
  objective value at most $b\cdot\lambda_G(\dem')$.
Then the graph $G'$ obtained from $G$ by merging 
each $S_i$ into a single vertex (keeping parallel edges%
\footnote{From our perspective of flows, 
parallel edges can also be merged into one edge with the same total capacity.
}%
) is a flow-sparsifier of $G$ with quality $(1+3\veps)b$.
\end{prop}

\begin{myproof}
  The graph $G'$ can equivalently be defined as taking $G$ and adding
  the edges $E_0 \eqdef \cup_{i = 1}^m \binom{S_i}{2}$, each with
  infinite capacity---merging all vertices in each $S_i$ is the same as
  adding these infinite capacity edges. Formally, let $G'=(V,E\cup
  E_0,c')$ where $c'(e)=c(e)$ if $e\in E$ and $c'(e)=\infty$ if $e\in
  E_0$. Then for any $\dem\in\R_+^{\binom{T}{2}}$, it is immediate that
  $\lambda_{G'}(\dem) \ge \lambda_G(\dem)$---every flow that is feasible
  in $G$ is also feasible in $G'$, even without shipping any flow on
  $E_0$.

  For the opposite direction, without loss of generality we may assume
  (by scaling) that $\lambda_G(\dem)=1$.  Let $\dem'\in \Ddisc$ be the
  demand vector obtained from $\dem$ in the construction of $\Ddisc$ (by
  zeroing small coordinates and rounding downwards to the nearest power
  of $(1+\veps)$). Clearly, $\dem'\le \dem$.

  First, we claim that $\lambda_G(\dem') < 1+3\veps$.  Indeed, assume to
  the contrary that $(1+3\veps) \dem'$ is feasible in $G$, i.e.,
  $(1+3\veps)\dem'\in \D(G)$.  Then the demand
  \[ 
  \dem'' 
  \eqdef (1-\veps)(1+3\veps)\dem' 
          + \sum_{i\in \binom{T}{2}} \veps/\tbinom{k}{2}\cdot L_i \mathbf{e}_i,
  \]
  is a convex combination of demands in $\D(G)$, and thus also
  $\dem''\in \D(G)$.  Observe that $\dem''>\dem$ (coordinate-wise),
  because each coordinate of $\dem'$ was obtained from $\dem$ by
  rounding down and possible zeroing (if it is smaller than some
  threshold), but we more than compensate for this when $\dem''$
  is created by multiplying $\dem'$ by $(1-\veps)(1+3\veps)\ge 1+\veps$
  and adding more than the threshold.  By down-monotonicity of $\D(G)$
  we obtain that $\lambda_G(\dem)>1$ in contradiction to our
  assumption, and the claim that $\lambda_G(\dem') < 1+3\veps$ follows.

  To get a handle on the value $\lambda_{G'}(\dem)$, we
  rewrite~(\ref{eq:dualLP1}) for $G' = G + E_0$ to obtain LP \eqref{eq:dualLPstrong}.
\begin{equation} \label{eq:dualLPstrong}
\framebox{ $
\begin{array}{lllll}
  \lambda_{G'}(\dem) = 
  & \min 
  &  \sum_{e\in E} c_e \len_e 
  \\
  & \operatorname{s.t.}
  &  \sum_{\aset{s,t}\in\tbinom{T}{2}} d_{st}\cdot \dist_{st} \ge 1 
  \\
  & &  \len_{e}=0 
  & \forall e\in E_0\eqdef\cup_{i\in[m]} \tbinom{S_i}{2},
  \\
  & &  \dist_{st} \ge \sum_{e\in P} \len_e 
  & \forall \aset{s,t}\in\tbinom{T}{2} \text{ and $s$-$t$ path $P$ on $E\cup E_0$}
  \\
  & & \len_e \ge 0 
  & \forall e\in E
  \\
  & & \dist_{st} \ge 0 
  & \forall \aset{s,t}\in\tbinom{T}{2}.
\end{array}
$ } \tag{LP2}
\end{equation}

By our premise, the demand $\dem' \in \Ddisc$ has a $\Pi$-respecting
feasible solution $\aset{\dist_{st}, \len_e}$ with value at most
$b\cdot\lambda_G(\dem')$; note that such a $\Pi$-respecting distance
function is also a solution to~\eqref{eq:dualLPstrong}.
Hence $\lambda_{G'}(\dem') \leq b\cdot\lambda_G(\dem')$.
Plugging in $\lambda_{G}(\dem') \leq (1+3\veps)$ and the normalization $\lambda_G(\dem)=1$, 
we conclude that $\lambda_{G'}(\dem) \le (1+3\veps)b \lambda_G(\dem)$,
which completes the proof of Proposition \ref{prop:SuffCond2}. 
\end{myproof}

The next proposition is similar in spirit to the previous one,
but with the crucial difference that it allows (or assumes)
a different partition of $V$ for every demand $\dem \in\Ddisc$.
Its proof is a simple application of the Proposition~\ref{prop:SuffCond2}.

\begin{prop} \label{prop:SuffCond3} Let $G=(V,E,c)$ be a $k$-terminal
  network, and fix $0<\veps<1/3$, $b\ge1$, and $m\ge 1$.  Suppose that
  for every $\dem\in\Ddisc$, there is an $m$-way partition $\Pi_{\dem} =
  \{S_1, \ldots, S_{m}\}$ (with some sets $S_j$ potentially empty) 
  and a $\Pi_{\dem}$-respecting distance function that is a feasible
  solution to~(\ref{eq:dualLP1}) with objective value at most
  $b\cdot\lambda_G(\dem)$.
  Then $G$ has a flow-sparsifier $G'$ with quality
  $(1+3\veps)b$, which has 
  $$
  \card{V(G')} \leq m^{\card{\Ddisc}} \leq m^{(\veps^{-1}\log k)^{k^2}}
  $$ 
  vertices.  Moreover, this graph $G'$ is obtained by merging vertices
  in $G$.
\end{prop}

\begin{myproof}
  For every demand $\dem\in\Ddisc$, we know there is an appropriate
  $m$-way partition $\Pi_{\dem}$ of $V$.  Imposing all these partitions
  simultaneously yields a ``refined'' partition $\Pi = \{S'_1, \ldots,
  S'_{m'}\}$ in which the number of parts is $m'\le m^{\card{\Ddisc}}$,
  and two vertices in the same part $S'_j$ of this refined partition if
  and only if they are in the same part of every initial partition.  Now
  apply Proposition~\ref{prop:SuffCond2} using this $m'$-way partition
  (using that any $\Pi_{\dem}$-respecting distance function is also a
  $\Pi$-respecting one), we obtain graph $G'$ that is a flow-sparsifier
  of $G$ and has at most $m'$ vertices.  Finally, we bound
  $\card{\Ddisc}$ using \eqref{eq:DdiscSize}.
\end{myproof}

\subsection{Quasi-Bipartite Graphs via Clumping}
\label{sec:qb-clump}

As a warm-up, we use Proposition~\ref{prop:SuffCond3} to construct a
graph sparsifier for quasi-bipartite graphs with quality $(1+\veps)$,
where the size of the sparsifier is only a function of $k$ and
$\veps$. Recall that a graph $G$ with terminals $T$ is
\emph{quasi-bipartite} if the non-terminals form an independent
set~\cite{RV99}. For this discussion, we assume that the terminals
form an independent set as well, by subdividing every terminal-terminal
edge---hence the graph is just bipartite.

\begin{theorem} \label{thm:bipartite}
Let $G=(V,E,c)$ be a quasi-bipartite $k$-terminal network, 
and let $\veps \in (0,\maxeps)$.
Then $G$ admits a quality $1+\veps$ flow-sparsifier $\hatG$ of size 
$\exp \{ O(k\log\tfrac{k}{\veps} (\tfrac{1}{\veps}\log k)^{k^2}) \}$.
\end{theorem}

\begin{myproof}
  To apply Proposition \ref{prop:SuffCond3}, the main challenge is to
  bound $m$, the number of parts in the partition.  To this end, fix a
  demand $\dem\in\Ddisc\subseteq\D(G)$, and let $\aset{\len_{e},
    \dist_{ij}}$ be an optimal solution for the linear program
  \eqref{eq:dualLP1}, hence its value is $\sum_{e\in E}
  c_e\len_e=\lambda_G(\dem)$.  We will modify this solution into another
  one, $\aset{\len'_{e}, \dist'_{ij}}$, that satisfies the desired conditions.
  This modification will be carried out in a couple steps, where we
  initially work only with lengths $\len_{uv}$ of edges $(u,v)\in E$,
  and eventually let $\dist'$ be the metric induced by shortest-path
  distances.

  For every $s,t\in T$ define the interval $\Gamma_{st}\eqdef[\veps d_{st},d_{st}]$, and let $\Gamma=\aset{0}\cup \Big(\cup_{s,t\in T}
  \Gamma_{st}\Big)$, and let $\Gamma^\veps$ contain $0$ and all powers of
  $(1+ \veps)$ that lie in $\Gamma$. The following claim provides
  structural information about a ``nice'' near-optimal solution.

  \begin{claim}
    \label{clm:qb-struct}
    Fix a non-terminal $v\in V\setminus T$, Then the edges between $v$
    and its neighbors set $N(v)\subset T$ (in $G$) admit edge lengths
    $\aset{\widehat \len_{vt}: t\in N(v)}$ that
    \begin{itemize} \compactify
    \item are dominated by $l$ (namely, $\forall t\in N(v),\ \widehat
      \len_{vt} \leq \len_{vt}$);
    \item use values only from $\Gamma^\veps$ (namely, $\forall t\in N(v),\
      \widehat \len_{vt}\in \Gamma^\veps$); and
    \item satisfy $\frac{1+\veps}{(1-\veps)^2}$-relaxed shortest distance constraints
      (namely, $\forall s,t\in T,\ \widehat \len_{sv}+\widehat \len_{vt}
      \ge \frac{(1-\veps)^2}{1+ \veps} \dist_{st}$).
\end{itemize}
\end{claim}
\begin{myproof}
  Let every edge length $\widehat \len_{vt}$ be defined as $\len_{vt}$
  rounded down to its nearest value from $\Gamma^\veps$.  The first two
  claimed properties then hold by construction.  For the third property,
  recall that $l$ is a feasible LP solution, thus $\len_{sv}+\len_{vt}
  \ge d_{st}$.  Assume without loss of generality that $\len_{sv} \leq
  \len_{vt}$.  If the large one $\len_{vt} \ge (1-\veps)^2 d_{st}$, the
  claim follows because also $\widehat \len_{vt} \ge
  \frac{(1-\veps)^2}{1+\veps} d_{st}$, regardless of whether $\len_{vt}$
  is smaller or bigger than $d_{st}$, where the extra term of
  $(1+\veps)$ comes from rounding down to the nearest power of
  $(1+\veps)$.  Otherwise, $\len_{vt} < (1 - \veps)^2 d_{st}$ hence the
  smaller one $\len_{sv} \ge \veps(2 - \veps) d_{st}$, and rounding down
  ensures $\widehat \len_{sv} \ge \frac{\veps(2 - \veps)}{1 + \veps}
  d_{st}$. The fact that $\veps < \maxeps$ means ${\veps(2 - \veps)}{1 +
    \veps} \geq 1$, so rounding down does not zero out $\widehat
  \len_{sv}$.  We conclude that the new lengths are at least
  $\frac{1}{1+\veps}$ times the old ones, namely $\widehat
  \len_{sv} \geq \frac{\len_{sv}}{1+\veps}$ and $\widehat \len_{vt} \geq
  \frac{\len_{vt}}{1+\veps}$.
  The claim follows.
\end{myproof}

We proceed with the proof of Theorem \ref{thm:bipartite}. Define new
edge lengths $\aset{\len'_e: e\in E}$ by applying the claim and scaling
edge lengths by $\frac{1+\veps}{(1 - \veps)^2}$, namely for every $v\in
V\setminus T$, and an adjacent $t\in T$, set $\len'_{vt}\eqdef
\frac{1+\veps}{(1 - \veps)^2}\, \widehat \len_{vt}$. This scaling and
the third property of Claim~\ref{clm:qb-struct} ensures that the
shortest-path distances (using edge
lengths $\len'_e$) between each pair of vertices $i,j$  is at least
$\dist_{ij}$.

Now partition the non-terminals into buckets, where two non-terminals
$u,v\in V\setminus T$ are in the same bucket if they ``agree'' about
each of their neighbors $t\in T$: either (i) they are both non-adjacent
to $t$, or (ii) they are both adjacent to $t$ and $\len'_{ut} =
\len'_{vt}$. Observe that this bucketing is indeed a well-defined
equivalence relation. Now for every $u,v\in V\setminus T$ that are the
same bucket, add an edge of length $\len'_{uv}=0$, and let $E'$ denote
this set of new edges. Let $\dist'$ be the shortest-path distances
according to these new edge-lengths. Observe that the shortest-path
distances between the terminals are unchanged by the addition of these
new zero-length edges, even though the distances between some
non-terminals have obviously changed. Hence $(\len'_e, \dist'_{ij})$ is
a feasible solution to~(\ref{eq:dualLP1}), with objective function value
at most $\frac{1+\veps}{(1 - \veps)^2} \lambda_G(\dem) \leq (1+5\veps) \lambda_G(\dem)$.

We define the equivalence classes of the bucketing above as the sets
$S_i$ in Proposition~\ref{prop:SuffCond3}.  Each bucket corresponds to a
``profile'' vector with $k$ coordinates that represent the lengths of
$k$ edges going to the $k$ terminals, if at all there is an edge to the
terminals.  Each coordinate of this profile vector is an element of
$\Gamma^\veps$ or it represents the corresponding edge does not
exist. It follows that the number of buckets (or profile vectors) is $ m
\leq \Big(\binom{k}{2} (\log_{1+\veps}\tfrac1\veps +3) \Big)^k \leq
(O(\tfrac{k^2}{\veps} \log\tfrac{1}{\veps}))^k \leq
(O(\tfrac{k}{\veps}))^{2k} \leq \exp \{ {O(k\log(k/\veps))} \} $.
The theorem follows by applying Proposition \ref{prop:SuffCond3},
which asserts the existence of a flow-sparsifier with 
$ m^{(\veps^{-1}\log k)^{k^2}} 
  \leq \exp \{ O(k\log\tfrac{k}{\veps} (\tfrac{1}{\veps}\log k)^{k^2}) \} 
$
vertices.
\end{myproof}

\section{The Splicing and Composition Techniques}
\label{sec:splicing-lemmas}

We say that a path is \emph{terminal-free} if all its internal vertices
are non-terminals. 
This terminology shall be used mostly for flow paths,
in which the paths' endpoints are certainly terminals.
The lemma below refers to two \emph{different} methods of routing 
a demand $\dem$ in a network $G$. 
The first method is the usual (and default) meaning,
where the demand is routed along arbitrary flow paths.
The second method is to route the demand along terminal-free flow paths,
and we will say this explicitly whenever we refer to this method.
We use a parameter $\rho\ge1$ to achieve greater generality, 
although the case $\rho=1$ conveys the main idea.

\begin{lemma}[Splicing Lemma]
  \label{lem:splice}
  Let $G_a$ and $G_b$ be two networks having the same set of terminals
  $T$, and fix $\rho\ge 1$. Suppose that whenever a demand $\dem$
  between terminals in $T$ can be routed in $G_a$ \emph{using terminal-free
  flow paths}, demand $\dem/\rho$ can be routed in $G_b$ (by arbitrary flow
  paths). Then for every demand $\dem$ between terminals in $T$ that can
  be routed in $G_a$, demand $\dem/\rho$ can be routed in $G_b$.
\end{lemma}

\begin{myproof} 
  Consider a demand $\dem$ that is routed in $G_a$ using flow $f^*$,
  and let us show that it can be routed also in $G_b$. 
  Fix for $f^*$ a flow decomposition $D = \{ (P_1, \phi(P_1)), (P_2,
  \phi(P_2)), \ldots\}$ for it, where each $P_l$ is a terminal-to-terminal 
  path, and $\phi(P_l)$ is the amount of flow sent on this path. 
  A flow decomposition also specifies the demand vector since
  $d_{st} = \sum_{\text{$st$-paths $P \in D$}} \phi(P)$. 
  If all the paths $(P,\phi)\in D$ are terminal-free, then we know by 
  the assumption of the lemma that demand $\dem/\rho$ can be routed in $G_b$. 
  Else, take a path $(P,\phi)\in D$ that contains internally some terminal---say 
  $P$ routes flow between terminals $t', t''$ and uses another terminal $s$ 
  internally. 
  We may assume without loss of generality that the flow paths are simple, 
  so $s\not\in \{t',t''\}$. 
  We replace the flow $(P, \phi(P))$ in $\dem$ by the two paths 
  $(P[t',s], \phi)$ and $(P[s,t''],\phi)$ to get a new flow decomposition $D'$, 
  and denote the corresponding demand vector by $\dem'$. 
  Note that $d'_{t',t''} = d_{t',t''} - \phi$, 
  whereas $d'_{t',s} = d_{t',s} + \phi$ and the same for $d'_{s,t''}$. 
  Moreover, if $\dem'/\rho$ can be routed on some graph $G_b$ with an arbitrary 
  routing, we can connect together $\phi/\rho$ amount of the flow from $t'$ to $s$ 
  with $\phi/\rho$ flow from $s$ to $t''$ to get a feasible routing for 
  $\dem/\rho$ in $G_b$. 
  Moreover the total number of terminals occurring internally on paths in 
  the flow decomposition $D'$ is less than that in $\dem$, 
  so the proof follows by a simple induction.
\end{myproof}

The next lemma addresses the case where our network 
can be described as the gluing of two networks $G_1$ and $G_2$,
and we already have sparsifiers for $G_1$ and $G_2$;
in this case, we can simply glue together the two sparsifiers, 
provided that the vertices at the gluing locations are themselves terminals.
Formally, let $G_1$ and $G_2$ be networks on disjoint sets of vertices, having
terminal sets $T_1 = \{s_1, s_2, \ldots, s_a\}$ and $T_2 = \{t_1, t_2,
\ldots, t_b\}$ respectively. Given a bijection $\phi \eqdef \{ s_1
\leftrightarrow t_1, \ldots, s_c \leftrightarrow t_c\}$ between some
subset of $T_1$ and $T_2$, the \emph{$\phi$-merge} of $G_1$ and $G_2$
(denoted $G_1 \oplus_\phi G_2$) is the graph formed by identifying
the vertices $s_i$ and $t_i$ for all $i \in [c]$. Note that the set of
terminals in $G$ is $T := T_1 \cup \{ t_{c+1}, \ldots, t_b\}$.

\begin{lemma}[Composition Lemma]
  \label{lem:compose}
  Suppose $G = G_1 \oplus_\phi G_2$. For $j \in \{1,2\}$, let $G'_j$ be
  a flow-sparsifier for $G_j$ with quality $\rho_j$. Then the graph $G'
  = G_1' \oplus_\phi G_2'$ is a quality $\maxx{\rho_1, \rho_2}$ flow
  sparsifier for $G$.
\end{lemma}

\begin{myproof}
  Consider a demand $\dem$ that is routable in $G$ using flow paths
  that do not have internal terminals. Since $G$ is formed by gluing
  $G_1$ and $G_2$ at terminals, this means each of the flow paths lies
  entirely within $G_1$ or $G_2$. 
  We can write $\dem=\dem_1+\dem_2$, where each $\dem_j$ is the demand 
  being routed on the flow paths among these that lie within $G_j$. 
  By the definition of flow-sparsifiers, these demands are also routable 
  in $G_1', G_2'$ respectively, and hence demand $\dem_1+\dem_2=\dem$ is
  routable in $G'$ (in fact by paths that lie entirely within $G_1$ or $G_2$). 
  Applying the Splicing Lemma (with $\rho=1$), we get that every demand 
  $\dem$ routable in $G$ is routable also in $G'$.

  The argument in the other direction is similar. Assume $\dem$ is
  routable in $G'$ using terminal-free flow paths; then we get two
  demands $\dem_1, \dem_2$ routable entirely in $G_1', G_2'$ respectively. 
  Scaling these demands down by $\maxx{\rho_1, \rho_2}$, they can be routed 
  in $G_1, G_2$ respectively, and hence we can route their sum 
  $(\dem_1+\dem_2)/\maxx{\rho_1,\rho_2}$ in $G$. 
  Applying the Splicing Lemma with $\rho=\maxx{\rho_1,\rho_2}$, we get
  a similar conclusion for all demands routable in $G'$ (on arbitrary 
  flow paths), and this completes the proof.
\end{myproof}

\paragraph{Applications of Splicing/Composition.}
The Splicing and Composition Lemmas will be useful in many of our
arguments: we use them to show a singly-exponential bound for
quasi-bipartite graphs in
Section~\ref{sec:qbip-via-splice} below,
in the sampling approach for quasi-bipartite graphs in
Section \ref{sec:sampling}, and also in constructing flow-sparsifiers
for series parallel and bounded treewidth graphs in
Section~\ref{sec:flowcuts}.

\subsection{Quasi-Bipartite Graphs via Splicing}
\label{sec:qbip-via-splice}

We show how to use Splicing Lemma \ref{lem:splice} to construct a flow
sparsifier for the quasi-bipartite graph of size $(1/\eps)^{\tilde
  O(k)}$.

\begin{theorem}
Let $G=(V,E,c)$ be a quasi-bipartite $k$-terminal network, 
and let $\veps \in (0,\maxeps)$.
Then $G$ admits a quality $1+\veps$ flow-sparsifier $\hatG$ of size 
$(1/\eps)^{\tilde O(k)}$.
\end{theorem}

\begin{myproof}
The construction goes through several stages. First, we construct $G'$
by rounding down the capacity to an integer power of $1+\eps$. The
main idea is to define ``types'' for non-terminals $v$ and then merge all
vertices of the same type (i.e., the new edge capacity is the sum of
the respective edge capacities incident to the merged vertices). The
main difficulty is in defining the types.

To define the type, first of all partition all non-terminals $v$ into
``super-types'', according to the set $S$ of terminals that are
connected to $v$ by edges with non-zero capacity. Now fix one such
super-type $S$, i.e., all vertices $v$ such that $\{t\in
T:c_{vt}\neq0\}=S$. Without loss of generality, suppose
$S=\{t_1,\ldots t_{h+1}\}$ and $c_{vt_i}\ge c_{vt_{i+1}}$ for
$i\in[h]$. For a vertex $v$, consider the vector of ratios
$r^*_v=\{c_{vt_1}/c_{vt_2}, c_{vt_2}/c_{vt_3},\ldots
c_{vt_h}/c_{vt_{h+1}})$. Note that $r^*_v$'s entries are all power of
$1+\eps$.  Now let $M=k^2/\eps+1$, and define $r_v$ by thresholding
all entries of $r^*_t$ exceeding $M$ by $M$. The $r_v$ defines the
type of the vertex $v$. Now we merge all vertices $v$ with the same
super-type $S$ and type $r_t$. Denote the new capacities $\hat
c_{t,u}$ for a terminal $t$ and a non-terminal node $u$ in $\hat G$.

Now we proceed to the analysis. First of all, notice that $G'$ is a
quality $1+\eps$ flow sparsifier, so we will care to preserve its
flows only. Furthermore, since the main operation is merging of the
nodes, we can only increase the set of feasible demands
in $\hat G$. The main challenge is to prove that if we can route a
demand vector $\dem$ in $G'$, we can route a demand $(1-O(\eps))\dem$
in $G'$. Using the Splicing Lemma \ref{lem:splice}, it is enough to
consider only demands $\dem$ that are feasible using 2-hop paths.

Fix some demand vector $\dem$ that is feasible in $\hat G$ using 2-hop
paths only. Fix a non-terminal node $u\in \hat G$, and let $f_{s,t}$
be the flow (of the solution) between $s,t\in T$ via $u$. Suppose $u$
has super-type $S$ and type $r=(r_1,\ldots r_h)$. We will show that we
can route $(1-O(\eps))f_{s,t}$ in $G$ for all $s$ to $t$ via the nodes
$v\in G'$ that have super-type $S$ and type $r$. This would clearly be
sufficient to conclude that $(1-O(\eps))\dem$ is feasible in $G'$.
Let $v_1,\ldots v_m$ be the nodes with super-type $S$ and type $r$.

We proceed in stages, routing iteratively from the ``small flows'' to
the ``large flows'' via $u$. Consider a suffix of $r$, denoted
$r_i,\ldots r_h$ where $r_i=M$ and $r_{i'}<M$ for all $i'>i$. For
$j\in[m]$, let $\alpha_{v_j}=c_{t_i,v_j}/\hat c_{t_i,u}$. Now for all
flows $f_{st}$, where $s\in\{t_{i+1},\ldots t_{h+1}\}$ and
$t\in\{t_1,\ldots t_h\}$, we route $(1-\eps)\alpha_{v_j} f_{st}$ flow
from $s$ to $t$ via $v_j$ in $G'$. We argue this is possible (even
when doing this for all $s,t$). Namely, consider any edge
$e=(t_{i'},v_j)$ for $t_{i'}\in\{t_{i+1},\ldots t_{h+1}\}$. The flow
accumulated on this edge is:
\begin{align*}
\sum_t (1-\eps)\alpha_{v_j}f_{t_{i'},t}
&=
(1-\eps)c_{t_i,v_j}/\hat c_{t_i,u}\cdot \sum_t f_{t_{i'},t}
\\
&\le
(1-\eps)c_{t_i,v_j}/\hat c_{t_i,u}\cdot \hat c_{t_{i'},u}.
\end{align*}
Note that $c_{t_iv_j}/c_{t_{i'},v_j}=r_i\cdot r_{i+1}\cdots
r_{i'-1}$, and similarly $\hat c_{t_i,u}/\hat c_{t_{i'},u}=r_i\cdot
r_{i+1}\cdots r_{i'-1}$. Hence the above formula is bounded by
$(1-\eps)c_{t_{i'},v_j}$, i.e., we satisfy the edge capacity (with a
$1-\eps$ slack, which will help later). Furthermore, we have routed
$\sum_j (1-\eps)\alpha_{v_j}f_{st}=(1-\eps)f_{st}$ flow for each
$s,t$.

We will repeat the above procedure for the next suffix of $r$ until we
are done routing flow $G'$. Note that we have at most $k$ such stages.

We need to mention one more aspect in the above argument --- what
happens to the flow that is contributed to edges $(t_{i'},v_j)$ where
$i'\le i$? The total contribution is at most $k/M\le \eps/k$ fraction
of the capacity (since $r_i=M$), which, over all (at most) $k$ stages
is still at most $\eps$ fraction of the edge capacity. Since we left a
slack of $\eps$ in the capacity for each edge in the above argument,
we still satisfy the capacity constraint overall for each edge.

Finally, to argue about the size of $\hat G$, note that there are only
$2^k$ super-types, and there are at most $O(k^2/\eps)^k$ possible
vectors $r$, and hence $\hat G$ has size at most $O(2^k\cdot
O(k^2/\eps)^k)=(1/\eps)^{\tilde O(k)}$.
\end{myproof}

\section{A Sampling Approach for Flow Sparsifiers} 
\label{sec:sampling}

In this section we develop our sampling approach to construct flow
sparsifiers. In particular, for quasi-bipartite graphs we construct 
in this method flow sparsifiers of size bounded by a polynomial in
$k/\eps$. 
This family includes the graphs for which a lower bound (for exact 
cut/flow sparsification) was proved in \cite{KR13}, 
and we further discuss how our construction extends to include also 
the graphs for which a lower bound was proved in \cite{KRTV12}.

\subsection{Preliminaries}
\label{sec:prelims-sampling}

We say that a random variable is \emph{deterministic} if it has variance $0$
(i.e., it attains one specific value with probability $1$).

\begin{theorem}[A Chernoff Variant] \label{thm:ChernoffPrime}
Let $X_1,\ldots,X_m\ge0$ be independent random variables, such that
each $X_i$ is either deterministic or $X_i\in [0,b]$,
and let $X=\sum_{i=1}^m X_i$. Then
\begin{align*}
  \Pr\Big[X \leq (1-\veps)\EX[X] \Big]
  & \leq e^{-\veps^2 \EX[X] /(2b)},
  & \forall \veps\in(0,1),
  \\
  \Pr\Big[X \geq (1+\veps)\EX[X] \Big]
  & \leq e^{-\veps^2 \EX[X] /(3b)},
  & \forall \veps\in(0,1).
\end{align*}
\end{theorem}
\begin{myproof}
First, replace every deterministic $X_i$ with multiple random variables
that are still deterministic but are all in the range $[0,b]$.
It suffices to prove the deviation bounds for the new summation,
because the new variables trivially maintain the independence condition,
and the deviation bound does not depend on the number $m$ of random variables.

Assuming now that every random variable is in the range $[0,b]$,
the deviation bounds follow from standard Chernoff bounds \cite{MR95,DP09}
by scaling all the random variables by factor $1/b$.
\end{myproof}

\subsection{Quasi-Bipartite Graphs}

Recall that a quasi-bipartite graph is one where the non-terminals form
an independent set; i.e., there are no edges between non-terminals.

\begin{theorem} \label{thm:bipartiteNew} Let $G=(V,E,c)$ be a
  quasi-bipartite $k$-terminal network, and let $0<\veps<\maxeps$.  Then $G$
  admits a quality $1+\veps$ flow-sparsifier $\hatG$ that has at most
  $\tilde O(k^7/\veps^3)$ vertices.
\end{theorem}

Our algorithm is randomized, and is based on importance sampling, as
follows.  Throughout, let $T\subset V$ be the set of $k$ terminals, and
assume the graph is connected.  We may assume without loss of generality
that $T$ also forms an independent set, by subdividing every edge that
connects two terminals (i.e., replacing it with a length $2$ path whose
edges have the same capacities as the edge being replaced).
We use a parameter $M\eqdef C\veps^{-3} k^5\log (\tfrac{1}{\veps}\log k)$,
where $C>0$ is a sufficiently large constant.
\begin{enumerate} \compactify
\item 
For every $s,t\in T$, compute a maximum $st$-flow in $G$ along 2-hops paths.
These path are edge-disjoint and each is identified by its middle vertex,
this flow is given by
\begin{equation} \label{eq:Fst}
  F_{st} \eqdef \sum_{v\in V\setminus T} F_{st,v},
  \quad \text{ where } \quad
  F_{st,v} \eqdef \minn{ c_{sv},c_{vt} }.
\end{equation}
\item 
For every non-terminal $v\in V\setminus T$, define a sampling probability
\begin{equation} \label{eq:pv}
  \tilp_v\eqdef \minn{1,p_v},
  \quad \text{ where } \quad
  p_v \eqdef M \cdot \max 
    \Big\{ \frac{F_{st,v}}{F_{st}}:\ s,t\in T \text{ and } F_{st,v}>0 \Big\}.
  \end{equation}
\item 
Sample each non-terminal with probability $\tilp_v$;
more precisely, for each $v\in V\setminus T$ independently at random,
with probability $\tilp_v$ scale the capacity of every edge incident to $v$ 
by a factor of $1/\tilp_v$, 
and with the remaining probability remove $v$ from the graph.
\item
Report the resulting graph $\hatG$.
\end{enumerate}

For the sake of analysis, it will be convenient to replace step 3 with 
the following step, which is obviously equivalent in terms of flow.
\begin{enumerate} \compactify
\setcounter{enumi}{2}
\renewcommand{\theenumi}{\arabic{enumi}'}
\item \label{it:3'}
For each $v\in V\setminus T$, set independently at random 
$I_v=1$ with probability $\tilp_v$ and $I_v=0$ otherwise (with probability $1-\tilp_v$),
and scale the capacities of every edge incident to $v$ by a factor of $I_v/\tilp_v$.
\end{enumerate}

We first bound the size of $\hatG$,
and then show that \whp\ $\hatG$ is a flow-sparsifier with quality $1+O(\veps)$.

\begin{lemma} \label{lem:samplingSize}
With probability at least $0.9$, 
the number of vertices in $\hatG$ is at most $O(k^2M)$.
\end{lemma}
\begin{myproof}
The number of vertices in $\hatG$ is exactly $\sum_{v\in V\setminus T} I_v$,
hence its expectation is
\[
  \EX\Big[\sum_{v\in V\setminus T} I_v\Big] 
  \leq \sum_{v\in V\setminus T} p_v
  \leq M \sum_{v\in V\setminus T}\ \sum_{s,t\in T:\ F_{st,v}>0} \tfrac{F_{st,v}}{F_{st}}
  = M \sum_{s,t\in T:\ F_{st}>0}\ \sum_{v\in V\setminus T} \tfrac{F_{st,v}}{F_{st}}
  \leq O(k^2 M),
\]  
where the second inequality simply bounds the maximum in \eqref{eq:pv}
with a summation,
and the last inequality follows from \eqref{eq:Fst}.
The lemma then follows by applying Markov's inequality.
\end{myproof}

\begin{lemma} \label{lem:samplingMain}
Let $d$ range over all nonzero demand vectors in $\R_+^{\binom{T}{2}}$. Then
\begin{align} 
  \Pr\Big[ \forall d\neq 0,\ \lambda_\hatG(d) \geq (1-3\veps) \lambda_G(d) \Big]
  \geq 0.9, \label{eq:samplingLB}
  \\
  \Pr\Big[ \forall d\neq 0,\ \lambda_\hatG(d) \leq (1+4\veps) \lambda_G(d) \Big]
  \geq 0.9. \label{eq:samplingUB}
\end{align}
\end{lemma}

Observe that Theorem \ref{thm:bipartiteNew} follows immediately from
Lemmas \ref{lem:samplingSize} and \ref{lem:samplingMain}.  It remains to
prove the latter lemma, and we do this next.  We remark that the $0.9$
probabilities above are arbitrary, and can be easily improved to be
$1-o(1)$.

\subsubsection{Proving the Lower Bound \eqref{eq:samplingLB}}

The plan for proving \eqref{eq:samplingLB} is to discretize the set of
all demand vectors, show a deviation bound for each of these demands (separately),
and then apply a union bound.  We will thus need the next lemma, which
shows that for every fixed demand vector $\dem$ that (satisfies some
technical conditions and) is feasible in $G$, 
\whp\ a slightly scaled demand $(1-\veps)\dem$ is feasible in $\hatG$.

Given a demand vector $\dem$, the problem of concurrent flow along 2-hop paths 
can be written as linear program \eqref{eq:samplingFlowLP}.
It has variables $f^{st}_v$ representing flow along a path $s-v-t$,
for the commodity $s,t\in T$ and intermediate non-terminal $v\in V\setminus T$.
Let $N_G(w)$ denote the set of neighbors of vertex $w$ in the graph $G$.
\begin{equation} \label{eq:samplingFlowLP}
\framebox{ $
\begin{array}{llll}
  \max & \lambda
  \\
  \operatorname{s.t.}
  & \sum_{v\in N_G(s)\cap N_G(t)} f^{st}_v \ge d_{st} \lambda \;
  & \forall \aset{s,t}\in\tbinom{T}{2}
  \\
  & \sum_{s\in N_G(v)\setminus\aset{t}} f^{st}_v \leq c_{vt}
  & \forall (v,t)\in E
  \\
  & f^{st}_v \ge 0 
  & \forall \aset{s,t}\in\tbinom{T}{2},\forall v\in N_G(s)\cap N_G(t).
\end{array}
$ } \tag{LP3}
\end{equation}

\begin{lemma} \label{lem:samplingLBfixed}
Fix $\eta>0$ and $\dem\in \R_+^{\binom{T}{2}}\setminus\aset{0}$ 
such that 
(i) demand $\dem$ can be satisfied in $G$ by flow along $2$-hop paths, and
(ii) every nonzero coordinate $d_{st}$ in $\dem$ is a power of $1+\veps$ 
in the range $[\eta F_{st},F_{st}]$.%
\footnote{The range upper bound $F_{st}$ follows anyway from requirement~(i).
We also remark that the requirement about power of $1+\veps$ 
is not necessary for the lemma's proof,
but for later use, it is convenient to include it here.
}
Then 
\begin{equation*}
  \Pr[ \text{demand $(1-\veps)\dem$ admits a flow in $\hatG$ along $2$-hop paths} ] 
    \geq 1 - \textstyle \binom{k}{2}\, e^{-\veps^2\eta M/2}.
\end{equation*}
\end{lemma}

\begin{myproof}
Given demand vector $\dem$, 
fix a flow $f$ that satisfies it in $G$ along $2$-hop paths.
Thus, $f_{st}=d_{st} \ge \eta F_{st}$.
Let $\hatG$ be the graph constructed using the above randomized procedure,
and recall that random variable $I_v$ is an indicator for the event
that non-terminal $v$ is sampled in step \ref{it:3'},
which happens independently with probability $\tilp_v$.

Define a flow $\hatf$ in $\hatG$ in the natural way:
scale every flow-path in $f$ whose intermediate vertex is 
$v\in V\setminus T$ by the corresponding $I_v/\tilp_v$. 
The resulting flow $\hatf$ is indeed feasible in $\hatG$ along $2$-hop paths.
It remains to prove that \whp\ this flow $\hatf$ routes at least
(i.e., a demand that dominates) $(1-\veps)\dem$.

Fix a demand pair (commodity) $s,t\in T$.
The amount of flow shipped by $\hatf$ along the path $s-v-t$ is
$\hatf^{st}_v \eqdef f^{st}_v I_v/\tilp_v$,
and the total amount shipped by $\hatf$ between $s$ and $t$
is 
\begin{equation} \label{eq:hatf}
  \hatf^{st} 
  \eqdef \sum_{v\in N_G(s)\cap N_G(t)} \hatf^{st}_v
  = \sum_{v\in N_G(s)\cap N_G(t)} f^{st}_v\cdot I_v/\tilp_v.
\end{equation}
By linearity of expectation,
$
  \EX[\hatf^{st}] 
  = \sum_v f^{st}_v\cdot \EX[I_v]/\tilp_v
  = \sum_v f^{st}_v
  = f^{st}
$.
Furthermore, we wrote $\hatf^{st}$ in \eqref{eq:hatf} 
as the sum of independent non-negative random variables, 
where each of summand is either deterministic (when $\tilp_v=1$),
or (when $\tilp_v=p_v<1$) can be bounded using \eqref{eq:pv} by
\[
  f^{st}_v\cdot I_v/\tilp_v 
  \leq f^{st}_v/\tilp_v 
  \leq F_{st,v}/p_v 
  \leq F_{st}/M.
\]
Applying Theorem \ref{thm:ChernoffPrime}, we obtain, as required,
\[
  \Pr[\hatf^{st} \le (1-\veps) f^{st}]
  \leq e^{-\veps^2 f^{st} / (2 F_{st}/M)}
  \leq e^{-\veps^2\eta M/2}.
\]
A straightforward union bound over the $\binom{k}{2}$ choices of $s,t$ 
completes the proof of Lemma~\ref{lem:samplingLBfixed}.
\end{myproof}

We proceed now to prove \eqref{eq:samplingLB} using Lemma \ref{lem:samplingLBfixed}. 
\begin{myproof}[Proof of Eqn.~\eqref{eq:samplingLB}]
Set $\eta\eqdef \veps/k^2$ and define
\begin{equation*}  %
  \D_{LB} \eqdef 
  \Big\{ \text{$\dem\in \R_+^{\binom{T}{2}}\setminus\aset{0}$ that satisfy requirements 
      (i) and (ii) in Lemma \ref{lem:samplingLBfixed}} 
  \Big\}.
\end{equation*}
Then clearly
$
  \card{\D^{LB}} 
  \leq \Big(2+\log_{1+\veps} \tfrac1\eta \Big)^{\binom{k}{2}}
  \leq \Big(\tfrac1\veps \log \tfrac{k}{\veps} \Big)^{k^2}
  \leq \Big(\tfrac{\log k}{\veps}\Big)^{O(k^2)}.
$
Applying Lemma \ref{lem:samplingLBfixed} to each $d\in\D^{LB}$ and using a
union bound, we get that with probability at least
$1-\card{\D^{LB}}\cdot \binom{k}2 e^{-\veps^2 \eta M/2} \ge 0.9$, for every $d\in\D^{LB}$
we have that $(1-\veps)d$ can be satisfied in $\hatG$ by 2-hop flow
paths.  We assume henceforth this high-probability event indeed occurs,
and show how this assumption implies the event described in \eqref{eq:samplingLB}.

To this end, fix a demand vector $\dem\in
\R_+^{\binom{T}{2}}\setminus\aset{0}$, and let us prove that
$\lambda_\hatG(\dem) \geq (1-3\veps) \lambda_G(\dem)$. We can make two
simplifying assumptions about the demand vector $\dem$, both of which
are without loss of generality.  Firstly, we assume that
$\lambda_G(\dem)\geq 1$, i.e., demand $\dem$ can be satisfied in $G$, 
because event in~\eqref{eq:samplingLB} is invariant under scaling of $\dem$. 
Secondly, we assume that $\dem$ can be satisfied in $G$ by $2$-hop flow
paths; if each such demand $\dem$ can be satisfied in $\widehat{G}$ with
congestion at most $1/(1-3\veps)$, then Lemma~\ref{lem:splice} implies 
that every demand satisfiable in $G$ (without the restriction to $2$-hop paths)
can be satisfied in $\widehat{G}$ with the same congestion.

So consider a demand $\dem \in \R_+^{\binom{T}{2}}\setminus\aset{0}$,
such that $\lambda_G(\dem)=1$ and $\dem$ can be satisfied in $G$ by
2-hop flow paths.  Let $\dem^-$ be the vector obtained from $\dem$ by
zeroing every coordinate $d_{st}$ that is smaller than $2\eta\, F_{st}$.
This vector can be written as $\dem^-=\dem-\sum_{st} \alpha_{st}
\mathbf{e}_{st}$, where $\mathbf{e}_{st}\in\R_+^{\binom{T}{2}}$ is the
standard basis vector for pair $(s,t)$, and $\alpha_{st}\eqdef d_{st}$
if this value is smaller than $2\eta F_{st}$, and zero
otherwise. Rounding each nonzero coordinate of $\dem^-$ down to the next
power of $1+\veps$ yields a demand vector $\dem^{LB}\in \D^{LB}$, and thus by
our earlier assumption, $(1-\veps)\dem^{LB} \geq \frac{1-\veps}{1+\veps}\dem^-$
can be satisfied in $\hatG$ by 2-hop flow paths.
For each $s,t\in T$, consider the demand vector $F_{st}\mathbf{e}_{st}$.
By rounding its single nonzero coordinate down to the next power of
$1+\veps$, we obtain a vector in $\D^{LB}$. Hence we conclude that
$\frac{1-\veps}{1+\veps} F_{st}\mathbf{e}_{st}$ can be satisfied in
$\hatG$ by 2-hop flow paths.  The set of demands satisfiable by 2-hop
flow paths in $\hatG$ is clearly convex, so taking a combination of such
demand vectors with coefficients that add up to $(1-\veps) +
\sum_{s,t\in T}\frac{\alpha_{st}}{F_{st}} \leq (1-\veps) +
\tfrac{k^2}{2}\cdot 2\eta = 1$, we conclude that
\[
  (1-\veps)\cdot \tfrac{1-\veps}{1+\veps}\,\dem^- 
  + \sum_{st}\frac{\alpha_{st}}{F_{st}} \cdot \tfrac{1-\veps}{1+\veps} F_{st}\mathbf{e}_{st}
  \geq \tfrac{(1-\veps)^2}{1+\veps} \Big[ \dem^- + \sum_{st} \alpha_{st} \mathbf{e}_{st}\Big]
  \geq (1-3\veps) \dem
\]
can be satisfied in $\hatG$.
This implies that $\lambda_\hatG(\dem) \ge 1-3\veps$, 
which completes the proof of \eqref{eq:samplingLB}.
\end{myproof}

\subsubsection{Proving the Upper Bound \eqref{eq:samplingUB}}

The plan for proving \eqref{eq:samplingUB} is similar, i.e., to prove a
deviation bound for every demand in a small discrete set and then apply
a union bound.  However, we need to bound the deviation in the opposite
direction, and thus use the LP that is dual to flow (which can be viewed as
``fractional cut'').  
We will need a statement of the following form: 
for every fixed demand vector $\dem$ that (satisfies some technical
conditions and) is \emph{not} feasible in $G$, 
\whp\ the slightly further scaled-up demand
$(1+\veps)\dem$ is \emph{not} feasible in $\hatG$.
The next lemma proves such a statement, 
except that it considers only flow along $2$-hop paths,
and that $\dem$ is scaled by another $1+\veps$ factor.

\begin{lemma} \label{lem:samplingUBfixedDemand}
Fix $\eta>0$ and let $\dem\in \R_+^{\binom{T}{2}}\setminus\aset{0}$ 
be a demand vector such that 
(i) demand $(1+\veps)\dem$ cannot be satisfied in $G$ by flow along 2-hop paths, 
and 
(ii) every nonzero coordinate $d_{st}$ is a power of $1+\veps$ 
in the range $[\eta F_{st},F_{st}]$.%
\footnote{Again, the power of $1+\veps$ requirement is not really necessary to 
prove the lemma, and will be needed only later.
But by introducing it right now, we avoid having two versions of condition (ii).
}
Then 
$$
  \Pr[ \text{demand $(1+\veps)^2\dem$ admits a flow in $\hatG$ along $2$-hop flow paths} ] 
    \leq e^{-\veps^2 \eta M/k^2}.
$$  
\end{lemma}

Our proof of Lemma \ref{lem:samplingUBfixedDemand} uses LP duality for
flows along 2-hop paths, which we discuss first.  Recall that for a
given demand vector $\dem$, our linear program \eqref{eq:samplingFlowLP} describes
the problem of maximizing concurrent flow along $2$-hop paths.
Its dual LP, written below, has variables $\len_e$
representing the lengths of edges $e\in E$, 
and variables $y_{st}$ representing the distance
(along the shortest 2-hop path) between $s,t\in T$.
\begin{equation} \label{eq:samplingDualLP}
\framebox{ $
\begin{array}{llll}
  \min 
  &  \sum_{v\in V\setminus T}\sum_{t\in N_G(v)} c_{vt}\len_{vt} %
  \\
  \operatorname{s.t.}
  &  \sum_{{s,t}\in\tbinom{T}{2}} d_{st} y_{st} \ge 1 
  \\
  &  y_{st} \le \len_{sv}+\len_{vt} 
  & \forall \aset{s,t}\in\tbinom{T}{2},\ \forall v\in N_G(s)\cap N_G(t)
  \\
  & \len_e \ge 0 
  & \forall e\in E
  \\
  & y_{st} \ge 0 
  & \forall \aset{s,t}\in\tbinom{T}{2}.
\end{array}
$ } \tag{LP4}
\end{equation}
By strong LP
duality, (\ref{eq:samplingDualLP}) has the same value
as~(\ref{eq:samplingFlowLP}) (assuming the primal LP is feasible and
bounded, which happens if, for every demand $d_{st}>0$, there is a
non-terminal $v\in V\setminus T$ connected to both $s$ and $t$ with
edges of positive capacity).

We can use these two LPs to reinterpret our algorithm's sampling probabilities,
namely the values $F_{st}$ and $p_v$ computed in \eqref{eq:Fst} and \eqref{eq:pv}.
(These will be needed in the proof of Lemma \ref{lem:samplingUBfixedDemand}.)
Consider a demand vector $\dem=\mathbf{e}_{s't'}$ for some fixed $s',t'\in T$, 
i.e., a unit demand for commodity $\aset{s',t'}$ and zero otherwise.
We shall assume there is $v\in V\setminus T$ that is connected 
to both $s'$ and $t'$ with edges of positive capacity. 
The next two lemmas analyze the optimal solutions to the two LPs above
for this demand vector.

\begin{lemma} \label{lem:samplingFlowLPsingle}
Fix a demand vector $\dem=\mathbf{e}_{s't'}$.
Then LP \eqref{eq:samplingFlowLP} has an optimal solution with
$f^{s't'}_v\eqdef F_{s't',v}$, all other flows are $0$, and $\lambda\eqdef F_{s't'}$.
\end{lemma}
\begin{proof}
Immediate from the fact that $2$-hop flow paths are edge-disjoint, 
as explained in \eqref{eq:Fst}.
\end{proof}

\begin{lemma} \label{lem:samplingDualLPsingle}
Fix a demand vector $\dem=\mathbf{e}_{s't'}$.
Then LP \eqref{eq:samplingDualLP} has an optimal solution 
$\aset{\len^{s't'}_e}_{e\in E},\aset{y^{s't'}_{st}}_{s,t\in T}$ 
where every non-terminal $v\in V\setminus T$ contributes to the objective
$\sum_{t \in N_G(v)} c_{vt}\len^{s't'}_{vt} = F_{s't',v}$. 
\end{lemma}
\begin{proof} 
Let us construct a solution to LP \eqref{eq:samplingDualLP},
denoted $\aset{\len_e},\aset{y_{st}}$ 
(we omit the superscript in this proof to simplify notation).
Let all edges $e$ not incident to either $s'$ or $t'$ have length $\len_e=0$,
and let $y_{st}=0$ for all $\aset{s,t}\neq \aset{s',t'}$.
Let $y_{s't'}=1$ and for every $v\in N_G(s')\cap N_G(t')$, 
let one of the two edges $(s',v)$ and $(v,t')$,
namely the one of cheaper cost
have length $1$, and the other one have length $0$ (breaking ties arbitrarily).
It is easy to verify that this is a feasible solution,
and every non-terminal $v$ contributes to the objective
$
  \sum_{t \in N_G(v)} c_{vt}\len_{vt} 
  = \sum_{t\in \aset{s',t'}} c_{vt}\len_{vt} 
  = \minn{c_{s'v},c_{vt'}} 
  = F_{s't',v}
$. 
Furthermore, the value of this solution is 
$\sum_{v\in V\setminus T} F_{s't',v} = F_{s't'}$.

Observing that the optimal LP value must be at least $F_{s't'}$ 
because of weak LP duality and Lemma \ref{lem:samplingFlowLPsingle},
we conclude that the constructed solution is indeed an optimal one.
\end{proof}

\begin{proof}[Proof of Lemma \ref{lem:samplingUBfixedDemand}]
Fix $\eta>0$ and let $\dem\in \R_+^{\binom{T}{2}}\setminus\aset{0}$
be a demand vector satisfying the two requirements.
We may assume that 
\begin{equation} \label{eq:dst_Fst}
  \forall s,t\in T,\ 
  \text{if $d_{st}>0$ then $F_{st}>0$},
\end{equation}
as otherwise the demand cannot be satisfied and the lemma's assertion holds trivially (the probability is $0$).
By requirement (i), the value of LP \eqref{eq:samplingFlowLP}, 
and thus also of LP \eqref{eq:samplingDualLP}, 
is smaller than $1+\veps$. 
Fix an optimal solution $\aset{\vec \ell, \vec y}$ for the latter LP; %
we can then write its value as 
\[
  z
  \eqdef \sum_{e\in E} c_e\len_e 
  = \sum_{v\in V\setminus T} \sum_{t\in N_G(v)} c_{vt}\len_{vt}
  < 1+\veps.
\]
In addition, the first constraint is tight, i.e., $\sum_{st} d_{st}y_{st}=1$,
as otherwise we can scale the entire solution to obtain a strictly better one.
(This holds for every optimal solution for every demand vector.)
Now consider the same values $\aset{\len_e},\aset{y_{st}}$ as the LP solution
for the graph $\hatG$ and same demand $\dem$,
where we use the viewpoint of step~\ref{it:3'} according to 
which $\hatG$ has the same edges as $G$
but the capacities of edges incident to every $v\in V\setminus T$ 
are scaled by $I_v/\tilp_v$.
This LP solution is obviously feasible also for $\hatG$,
and what remains is to prove a deviation bound on its objective value
\begin{equation}\label{eq:hatZ}
  \widehat Z 
  \eqdef \sum_{v\in V\setminus T}\sum_{t \in N_G(v)} (c_{vt}I_v/\tilp_v) \len_{vt}
  = \sum_{v\in V\setminus T}(I_v/\tilp_v) \sum_{t \in N_G(v)} c_{vt}\len_{vt}.
\end{equation}
By construction $\EX[I_v/\tilp_v]=1$, hence 
$\EX[\widehat Z] = z < 1+\veps$.
To prove a deviation bound on $\widehat Z$ using concentration from
Theorem~\ref{thm:ChernoffPrime}, we need an upper bound on each term
of the summation over $v$'s.  For this, we analyze how each sampling
probabilities $\tilp_v$ (which are set without ``knowing'' the demand
vector $\dem$) relate to the potential contributions $\sum_{t \in
  N_G(v)} c_{vt}\len_{vt}$ (which depend on $\dem$).  The key insight
is captured by the following claim.

\begin{claim} \label{cl:samplingUBcontrib} 
For every non-terminal $v\in V\setminus T$ with $\tilp_v<1$,
its maximum contribution to $z$ is 
$\sum_{t \in N_G(v)} c_{vt}\len_{vt} \le \tilp_v\cdot k^2/(2\eta M)$.
\end{claim}

\begin{proof}[Proof of Claim \ref{cl:samplingUBcontrib}]
Fix $v$ with $\tilp_v<1$, which implies $p_v=\tilp_v<1$.
The plan is to modify the optimal LP solution $\aset{\len_e},\aset{y_{st}}$,
by assigning new lengths to just the edges incident to $v$, and keeping
the old length assignments for the other edges.
Once we verify that the modified solution is feasible,
this modified solution will give us an upper bound on 
$v$'s contribution to the objective in the optimal LP solution.

We set the new edge lengths $\aset{\widetilde \len_e}$ as follows.
Consider a demand $\aset{s',t'}\in\binom{T}{2}$ such that $d_{s't'}>0$,
which implies $d_{s't'}\in [\eta F_{s't'},F_{s't'}]$;
moreover, $F_{s't'}>0$, by~(\ref{eq:dst_Fst}).
Let $\aset{\len^{s't'}_e}_{e\in E},\aset{y^{s't'}_{st}}_{s,t\in T}$ be 
an optimal LP solution for the single-commodity demand $\mathbf{e}_{s't'}$, 
as computed in Lemma~\ref{lem:samplingDualLPsingle}.
Now scale the edge lengths in this solution by $1/(\eta F_{s't'})$,
and add up over all such $\aset{s',t'}$, 
to get the new length for edges incident to $v$.
Formally, for every edge $e$, let
\begin{equation} \label{eq:tildel}
  \widetilde \len_e \eqdef 
  \begin{cases}
    \sum_{s',t':\ d_{s't'}>0} \ \len^{s't'}_e/(\eta F_{s't'})\quad
    & \text{if $e$ is incident to $v$ (in $G$),} \\
    \len_e
    & \text{otherwise.}
  \end{cases}
\end{equation}

To verify this LP solution is feasible, we only need to check that
$y_{st} \leq \widetilde \len_{sv}+\widetilde \len_{vt}$ for all $s,t\in N(v)$.
To this end, fix $s,t\in N(v)$.
We may assume that $d_{st}>0$, as otherwise 
$y_{st}$ can be set to a large enough value without affecting the objective
(strictly speaking, this modifies the LP also in some $y_{st}$
variables). 
We now have
\begin{align*}
  \widetilde \len_{sv}+\widetilde \len_{vt}
  & = \sum_{s',t':\ d_{s't'}>0} (\len^{s't'}_{sv} + \len^{s't'}_{vt}) / (\eta F_{s't'})
  & \text{by plugging \eqref{eq:tildel}}
  \\
  & \geq \sum_{s',t':\ d_{s't'}>0} y^{s't'}_{st} / (\eta F_{s't'})
  & \text{$\vec \len^{s't'},\vec y^{s't'}$ is feasible}
  \\ 
  & \geq y^{st}_{st} / (\eta F_{st})
  & \text{using $s'=s,\ t'=t$}.
\end{align*}
Moreover, the LP solution $\vec \len^{st}, \vec y^{st}$ 
(for the single-commodity demand $\mathbf{e}_{st}$),
satisfies the first constraint of LP, which simplifies to $y^{st}_{st}=1$.
Also the LP solution $\vec \ell, \vec y$ (for demand $\dem$) 
satisfies the first constraint, and then using requirement (ii), we have
$
  1 
  \geq d_{st} y_{st}
  \geq \eta F_{st} y_{st}
$.
Combining our last three estimates, we obtain
\[
  \widetilde \len_{sv}+\widetilde \len_{vt}
  \geq y^{st}_{st} / (\eta F_{st})
  = 1 / (\eta F_{st})
  \geq y_{st},
\]
which completes the verification that the modified LP solution is feasible.

The objective value $\sum_e c_e \len_e$ of the optimal LP solution is
clearly at most the objective value $\sum_e c_e \widetilde \len_e$ of
the modified LP solution, but since we only modified the length of edges
incident to $v$ (in $G$), we get
\begin{align*}
  \sum_\text{$e$ incident to $v$} c_{e}\len_{e} 
  & \le \sum_\text{$e$ incident to $v$} c_{e}\widetilde \len_{e} 
  \\
  & = \sum_\text{$e$ incident to $v$} c_e \cdot \sum_{s',t':\ d_{s't'}>0} \ \len^{s't'}_e/(\eta F_{s't'})
  & \text{by plugging \eqref{eq:tildel}}
  \\
  & = \sum_{s',t':\ d_{s't'}>0} 1/(\eta F_{s't'}) \cdot \sum_\text{$e$ incident to $v$} c_e \ \len^{s't'}_e
  & \text{interchanging summations}
  \\ 
  & = \sum_{s',t':\ d_{s't'}>0} 1/(\eta F_{s't'}) \cdot F_{s't',v}
  & \text{by lemma~\ref{lem:samplingDualLPsingle}}
  \\
  & \leq \sum_{s',t':\ d_{s't'}>0} p_v/(\eta M)
  & \text{by \eqref{eq:pv}}.
\end{align*}
The claim now follows by recalling that $p_v=\tilp_v$ 
and $\card{\binom{T}{2}}\leq k^2/2$.
\end{proof}

We can now continue with the proof of Lemma \ref{lem:samplingUBfixedDemand}
and prove the desired deviation bound on $\widehat Z$.
Recalling \eqref{eq:hatZ}, we can write $\widehat Z=\sum_{v\in
  V\setminus T} \widehat Z_v$;
each non-negative random variable $\widehat Z_v$ is either deterministic if $\tilp_v=1$, or
else $\tilp_v<1$, in which case we apply Claim
\ref{cl:samplingUBcontrib} to get the upper bound 
\[
  \widehat Z_v 
  \leq (1/\tilp_v) \sum_{t \in N_G(v)} c_{vt}\len_{vt}
  \leq k^2/(2\eta M).
\]
Applying Theorem \ref{thm:ChernoffPrime} 
and recalling that $\EX[\widehat Z] = z < 1+\veps$, we have
$
  \Pr[\widehat Z \geq (1+\veps)^2] 
  \leq e^{-\veps^2 (1+\veps)\eta M/k^2}
$,
which completes the proof of Lemma \ref{lem:samplingUBfixedDemand}.
\end{proof}

\begin{proof}[Proof of Eqn.~\eqref{eq:samplingUB}]
The proof generally resembles that of \eqref{eq:samplingLB},
although several details are different and somewhat more complicated.
Set $\eta\eqdef \veps/k^2$ and define
\begin{equation}  \label{eq:samplingUBDdisc}
  \D_{UB} \eqdef 
  \Big\{ \text{$\dem\in \R_+^{\binom{T}{2}}\setminus\aset{0}$ that satisfy requirements 
      (i) and (ii) in Lemma \ref{lem:samplingUBfixedDemand}} 
  \Big\}.
\end{equation}
Then clearly
$
  \card{\D^{UB}} 
  \leq \Big(2+\log_{1+\veps} \tfrac1\eta \Big)^{\binom{k}{2}}
  \leq \Big(\tfrac1\veps \log \tfrac{k}{\veps} \Big)^{k^2}
  \leq \Big(\tfrac{\log k}{\veps}\Big)^{O(k^2)}.
$
Applying Lemma \ref{lem:samplingUBfixedDemand} to each $\dem\in\D^{UB}$ 
and a straightforward union bound, 
we see that with probability at least 
$1-\card{\D^{UB}}\cdot e^{-\veps^2 \eta M/k^2}\ge 0.9$,
for every $\dem \in\D^{UB}$ we have that demand $(1+\veps)^2 \, \dem$ 
cannot be satisfied in $\hatG$ by flow along $2$-hop paths.
We assume henceforth that this high-probability event indeed occurs,
and show how this assumption implies the event in \eqref{eq:samplingUB}.

We thus aim to show that for every $\dem\in \R_+^{\binom{T}{2}}\setminus\aset{0}$ 
we have $\lambda_\hatG(\dem)\leq (1+4\veps)\lambda_G(\dem)$.
By scaling $\dem$ appropriately, 
it suffices to show that whenever $\lambda_\hatG(\dem) \geq 1$,
i.e., the demand $\dem$ can be satisfied in $\hatG$,
the slightly scaled demand $\tfrac{1}{1+4\veps}\ \dem$ can be satisfied in $G$.
By Lemma~\ref{lem:splice}, it suffices to prove a statement that is similar, 
but with the stronger hypothesis
that $\dem$ can be satisfied in $\hatG$ \emph{along $2$-hop paths}.
And indeed, this is what we prove next by way of contradiction.

Suppose, for sake of a contradiction, there is a demand $\dem\neq 0$ 
that can be satisfied in $\hatG$ \emph{along $2$-hop paths},
but $\tfrac{1}{1+4\veps}\dem$ cannot be satisfied in $G$,
i.e., $\lambda_G(\dem) < \tfrac{1}{1+4\veps}$.
Let $\dem^-$ be the vector obtained from $\dem$ by zeroing every
coordinate $d_{st}$ that is smaller than $2\eta F_{st}$.  We can write
this vector as $\dem^-=\dem - \sum_{st} \alpha_{st}\mathbf{e}_{st}$,
where $\mathbf{e}_{st}$ is the standard basis vector for pair $(s,t)$,
and $\alpha_{st}\eqdef d_{st}$ if this value is smaller than $2\eta
F_{st}$, and otherwise $\alpha_{st}\eqdef 0$.  Round each nonzero
coordinate of $\dem^-$ down to the next power of $1+\veps$, to obtain a
demand vector $\dem^{UB} \geq \tfrac{1}{1+\veps}\dem^-$.

\begin{claim} \label{cl:samplingUBinDdisc} 
$\dem^{UB}/(1+\veps)^2\in \D^{UB}$.
\end{claim}

\begin{proof}[Proof of Claim \ref{cl:samplingUBinDdisc}]
  We need to show that $\dem^{UB}/(1+\veps)^2$ satisfies the two conditions of
  Lemma \ref{lem:samplingUBfixedDemand}.  
  Starting with the proof of condition~(i) by way of contradiction,
  let us assume that demand $\dem^{UB}/(1+\veps)$ can be
  satisfied in $G$ by flow along $2$-hop paths.  The set of demands
  satisfiable in this manner (by $2$-hop flow paths in $G$) is convex
  and down-monotone, and by definition contains also the demands $F_{st}
  \mathbf{e}_{st}$ for every $s,t\in T$.  Thus, a linear combination of
  vectors in the set, whose coefficients are non-negative and sum up to
  $(1-\veps) + \sum_{s,t\in T}\frac{\alpha_{st}}{F_{st}} \leq (1-\veps)
  + \tfrac{k^2}{2}\cdot 2\eta = 1$, must also be in the set.  Taking
  the linear combination
\[
  (1-\veps)\cdot \frac{\dem^{UB}}{1+\veps} + \sum_{st\in T} \frac{\alpha_{st}}{F_{st}}\cdot F_{st} \mathbf{e}_{st}
  \geq \frac{1-\veps}{(1+\veps)^2}\, \dem^- + \sum_{st\in T} \alpha_{st} \mathbf{e}_{st}
  \geq \frac{1}{1+4\veps} \Big[ \dem^- + \sum_{st\in T} \alpha_{st} \mathbf{e}_{st} \Big]
  = \frac{1}{1+4\veps}\, \dem,
\] 
we see that $\frac{1}{1+4\veps}\,\dem$ can be satisfied in $G$ 
by flow along $2$-hop paths, 
and clearly also without the restriction on the flow paths.
The latter contradicts our earlier assumption that 
$\lambda_G(\dem) < \tfrac{1}{1+4\veps}$,
and thus proves condition~(i).

We now prove condition~(ii), which asserts that every nonzero coordinate 
$d'_{st}/(1+\veps)^2$ is in the range $[\eta F_{st},F_{st}]$.
One direction is immediate: if $d'_{st}$ is non-zero, then 
$
  d'_{st} 
  \geq d^{-}_{st}/(1+\veps) 
  \geq 2\eta F_{st}/(1+\veps) 
  \geq (1+\veps)^2\eta F_{st}
$.
For the other direction, observe that $F_{st}>0$ because otherwise
$G$ has no $2$-hop path of positive capacity between $s$ and $t$,
which implies the same in $\hatG$, and we get $d'_{st}\leq d_{st}=0$.
Define $\widetilde F_{st}$ to be $F_{st}$ rounded down 
to the next power of $1+\veps$, which means $(1+\veps)\widetilde F_{st} > F_{st}$.
Then the corresponding demand $\widetilde F_{st} \mathbf{e}_{st}\in \D^{UB}$,
and using our assumption about all demands in $\D^{UB}$,
demand $(1+\veps)^2\widetilde F_{st}$ cannot be satisfied by 
flow along $2$-hop paths in $\hatG$. 
Recalling that demand $\dem$ can be satisfied in that manner,
we derive the other direction
$
  d'_{st}
  \leq d^{-}_{st}
  \leq d_{st}
  < (1+\veps)^2 \widetilde F_{st}
  \leq (1+\veps)^2 F_{st}
$.
This completes the proof of requirement (ii), and of Claim \ref{cl:samplingUBinDdisc}.
\end{proof}

We now finish the proof of Eqn.~\eqref{eq:samplingUB}.  Recall that we
assumed the high probability event in
Lemma~\ref{lem:samplingUBfixedDemand} occurred for every demand in
$\D^{UB}$. Using Claim \ref{cl:samplingUBinDdisc} we know that
$\dem^{UB}/(1+\veps)^2$ satisfies the properties of
Lemma~\ref{lem:samplingUBfixedDemand}, which implies that the scaled-up
demand $\dem^{UB}$ cannot be satisfied in $\hatG$ by flow along $2$-hop
paths. Since $\dem^{UB} \leq \dem^- \leq \dem$, also demand $\dem$ cannot be
satisfied in this manner, but this contradicts the choice of $\dem$.
This completes the proof of \eqref{eq:samplingUB}.
\end{proof}
\subsection{An Extension to More General Graphs}
\label{sec:small-graphs}

An extension of the techniques for quasi-bipartite graphs is to the
following case: let $G$ be a terminal network such that if we delete the
terminal set $T$ then each component of $G \setminus T$ has at most $w$
nodes in it. (The case of quasi-bipartite graphs is precisely when
$w=1$.) The sampling technique extends to this case; we now sketch the
ideas for the extension.

Let the vertex sets of the components in $G \setminus T$ be $V_1, V_2,
\ldots, V_l$, with each $\card{V_i} \leq w$. (Again, assume $T$ forms an
independent set.) For each $s,t \in T$, compute a max-flow that is
terminal-free, i.e., it only uses flow-paths that go from $s$ to $t$
using vertices within a single $V_i$, and does not contain terminals as
internal nodes. Let $F_{st, i}$ be the value of this flow using $V_i$,
and $F_{st} = \sum_i F_{st,i}$ be the value of the maximum $s$-$t$
terminal-free flow itself. Observe that $F_{st,i}$ also equals the value
of the $s$-$t$ min-cut within the graph $G[V_i\cup\aset{s,t}]$. Define $p_i := M \cdot
\maxx{ \frac{F_{st,i}}{F_{st}} : s,t \in T }$, and the sampling
probability of component $i$ is then $\tilp_i := \min\{p_i, 1\}$. Now we
sample each component $i$ (i.e., keep subset $V_i$) independently
with probability $\tilp_i$, in which case we scale the capacities
of its incident edges by $1/\tilp_i$, to get overall a graph $\hatG$. 

The analysis proceeds almost unchanged.  The number of vertices in
$\hatG$ is now $O(k^2 M \cdot w)$ with high probability. The proof of the
lower bound~(\ref{eq:samplingLB}) is unchanged, apart from replacing the
use of $2$-hop paths by terminal-free paths. For the upper bound, we
again write down the dual LP for terminal-free flows (analogous
to~(\ref{eq:samplingDualLP})), construct dual solutions $\{
\len^{s't'}_e \}, \{y^{s't'}_{st} \}$ using the max-flow/min-cut duality
(as in Lemma~\ref{lem:samplingDualLPsingle}), and argue that the
contribution of each component $i$ to the LP value is bounded (as in
Claim~\ref{cl:samplingUBcontrib}). The rest of the arguments in the
upper bound proceed analogously to the quasi-bipartite case; 
details omitted.

\ignore{
\section{Directions}

\begin{itemize}
\item Can we get a smaller single-source flow sparsifier? Using
  sampling, say?

\item (Hard?) Can we show any $O(1)$-terminal network where we need
  $\omega(1)$ nodes for exact flow-sparsification?

\item Simpler question: for a tree we can just keep $k-1$ Steiner nodes
  and keep an exact flow sparsifier. Can we show that we cannot do exact
  flow sparsification using $o(k)$ Steiner nodes? What can we get for a
  tree with $(1+\varepsilon)$-flow sparsification?

\end{itemize}
} %

\section{Results Using Flow/Cut Gaps}
\label{sec:flowcuts}

Given the $k$-terminal network and demand matrix $\dem\neq 0$, recall that
$\lambda_G(\dem)$ is the maximum multiple of $\dem$ that can be sent
through $G$. We can also define the \emph{sparsity} of a cut $(S, V
\setminus S)$ to be 
\[ \Phi_G(S;\dem) \eqdef \frac{ \sum_{e \in \partial S}
  c_e }{ \sum_{i,j: |\{i,j\} \cap S| = 1} d_{ij} }, \] 
and the \emph{sparsest cut} as
\[ \Phi_G(\dem) \eqdef \min_{S \subseteq V} \Phi_G(S;\dem) . \] 
Define the \emph{flow-cut gap} as
\[ \gap(G) := \max_{\dem \in \D(G)} \Phi_G(\dem)/\lambda_G(\dem). \] It
is easy to see that $\Phi_G(\dem) \geq \lambda_G(\dem)$ for each demand
vector $\dem$ (and hence $\gap(G) \geq 1$); a celebrated result
of~\cite{LLR95,AR98} shows that for every $k$-terminal network $G$,
the gap is $\gap(G) \leq O(\log k)$. Many results are known about the flow-cut gap
based on the structure of the graph $G$ and that of the support of the
demands in $\D(G)$; in this section we use these results together with 
known results about cut sparsifiers, 
to derive new results about flow sparsifiers.

It will be convenient to generalize the notion of a $k$-terminal
network. Given a $k$-terminal network $G = (V,E,c)$ with its associated
subset $T \subseteq V$ of $k$ terminals, define the
\emph{demand-support} to be another undirected graph $H = (T,F)$ with some
subset of edges $F$ between the terminals $T$. The \emph{demand
  polytope} with respect to $(G,H)$ is the set of all demand vectors
$\dem = (d_e)_{e \in H}$ which are supported on the edges in the
demand-support $H$, that are routable in $G$; i.e.,
\begin{gather}
  \D(G,H) \eqdef \aset{ \dem \in \R_+^{F} :\ \text{demand $\dem$ can be
      routed in $G$} }. \label{eq:2}
\end{gather}
This is a generalization of~(\ref{eq:1}), where we defined $H$ to be the
complete graph on the terminal set~$T$. Define the flow-cut gap with
respect to the pair $(G,H)$ as
\[ \gap(G,H) := 
   \maxx { \Phi_G(\dem)/\lambda_G(\dem) :\ \dem \in \D(G,H)\setminus\aset{0}}. 
\] 

Analogously to a flow sparsifier, we can define cut-sparsifiers. Given a $k$-terminal network $G$ with
terminals $T$, a \emph{cut-sparsifier for $G$ with quality $\beta\geq 1$} is a
graph $G' = (V',E',c')$ with $T \subseteq V'$, such that for every
partition $(A,B)$ of $T$, we have
\[ \mincut_{G}(A,B) \leq \mincut_{G'}(A,B) \leq \beta \cdot
\mincut_{G}(A,B). \] A cut-sparsifier $G'$ is \emph{contraction-based}
if it is obtained from $G$ by increasing the capacity of some edges and
by identifying some vertices (from the perspective of cuts and flows, 
the latter is equivalent to adding infinite capacity edges between vertices).

\begin{theorem}
  \label{thm:cut-flow}
  Given a k-terminal network $G$ with terminals $T$, let $G'$ be a
  quality $\beta\ge1$ cut-sparsifier for $G$. Then for every
  demand-support $H$ and all $\dem \in \R_+^{E(H)}$,
  \begin{gather}
    \frac{1}{\gap(G',H)} \leq \frac{\lambda_{G'}(\dem)}{\lambda_G(\dem)}
    \leq \beta\cdot\gap(G,H). \label{eq:3}
  \end{gather}
  Therefore, the graph $G'$ with edge capacities scaled up by $\gap(G', H)$ 
  is a quality $\beta \cdot \gap(G,H) \cdot \gap(G',H)$
  flow sparsifier for $G$ for all demands supported on $H$.

  Moreover, if $G'$ is a contraction-based cut-sparsifier, then trivially
  $\lambda_G(\dem) \leq \lambda_{G'}(\dem)$,
  and hence $G'$ itself is a quality $\beta\cdot \gap(G,H)$
  flow sparsifier for $G$ for demands supported on $H$.
\end{theorem}

\begin{myproof}
  Consider a demand $\dem \in \D(G,H)$; the maximum multiple of it we
  can route is $\lambda_G(\dem)$. For any partition $(A,B)$ of the
  terminal set $T$, let $d(A,B) \eqdef \sum_{\aset{i,j}: |\aset{i,j} \cap A| = 1}
  d_{ij}$. 
  The flow across a cut cannot exceed that cut's capacity, 
  hence $\lambda_G(\dem)\cdot d(A,B) \leq \mincut_G(A, B)$.
  Since $G'$ is a cut sparsifier of $G$, we have
  $\mincut_{G}(A, B) \leq \mincut_{G'}(A, B)$, and together we obtain
  \[ \lambda_G(\dem) \leq \frac{\mincut_{G}(A, B)}{ d(A,B) }. \] 
  Minimizing the right-hand side over all partitions $(A,B)$ of the terminals, 
  we have $\lambda_G(\dem) \leq \Phi_{G'}(\dem)$. The flow-cut gap for $G'$
  implies that $\lambda_G(\dem) \leq \gap(G',H)\cdot\lambda_{G'}(\dem)$,
  which shows the first inequality in~(\ref{eq:3}). 
  For the second one, we just reverse the roles of $G$
  and $G'$ in the above argument, but now have to use that $
  \mincut_{G'}(A,B) \leq \beta \cdot \mincut_{G}(A,B)$ to get 
  $\frac{\lambda_{G'}(\dem)}{\beta} \leq \frac{\mincut_{G}(A, B)}{ d(A,B) }$, 
  and hence eventually that 
  $\lambda_{G'}(\dem) \leq \beta\cdot \gap(G,H)\cdot \lambda_{G}(\dem)$.

  For the second part of the theorem, observe that if $G'$ is
  contraction-based, then it is a better flow network than $G$, which
  means $\lambda_{G'}(\dem) \geq \lambda_{G}(\dem)$.
\end{myproof}

This immediately allows us to infer the following results.

\begin{corollary}[Single-Source Flow Sparsifiers]
  \label{cor:single-source}
  For every $k$-terminal network $G$, there exists a graph $G'$ with
  $2^{2^k}$ vertices that preserves (exactly) all single-source and
  two-source flows.%
  \footnote{A two-source flow means that there are two terminals $t',t''\in T$
    such that every non-zero demand is incident to at least one of $t',t''$. 
    Single-source flows are defined analogously with a single terminal.}
\end{corollary}

\begin{myproof}
  Hagerup et al.~\cite{HKNR98} show that all graphs have
  (contraction-based) cut-sparsifiers with quality $\beta=1$ and size
  $2^{2^{k}}$ 
  (see also \cite{KRTV12} for a slight improvement for undirected graphs). 
  Moreover, it is known that whenever $H$ has a vertex
  cover of size at most $2$, the flow-cut gap is exactly 
  $\gap(G,H)=\gap(G',H)=1$~\cite[Theorem~71.1c]{Sch-book}.
\end{myproof}

\begin{corollary}[Outerplanar Flow Sparsifiers]
  \label{cor:translate-planar}
  If $G$ is a planar graph where all terminals $T$ lie on the same face,
  then $G$ has an exact (quality $1$) flow sparsifier with $O(k^2 2^{2k})$ vertices.
  In the special case where $G$ is outerplanar, 
  the size bound improves to $O(k)$.
\end{corollary}

\begin{myproof}
  Okamura and Seymour~\cite{OS81} show that the flow-cut gap for
  planar graphs with all terminals on a single face is $\gap(G,H)=1$, and
  Krauthgamer and Rika~\cite{KR13} show that every planar graph $G$
  has a contraction-based cut-sparsifier $G'$ with quality $\beta=1$ and size
  $O(k^2 2^{2k})$. And since the latter is contraction-based, also this $G'$ 
  is planar with all terminals on a single face, hence $\gap(G',H)=1$.

  To improve the bound when $G$ is outerplanar, we use a result of 
  Chaudhuri \etal~\cite[Theorem 5(ii)]{CSWZ00} that every outerplanar
  graph $G$ has a cut-sparsifiers $G'$ with quality $\beta=1$ and size $O(k)$,
  and moreover, this also $G'$ is outerplanar and thus $\gap(G',H)=1$.
\end{myproof}

\begin{corollary}[$4$-terminal Flow Sparsifiers]
  \label{cor:translate-K4}
  For $k \leq 4$, every $k$-terminal network has an exact (quality $1$)
  flow sparsifier with at most $k+1$ vertices.
\end{corollary}

\begin{myproof}
  Lomonosov~\cite{Lomonosov85} shows that the flow-cut gap for at most $4$
  terminals is $\gap(G,H)=\gap(G',H)=1$, 
  and Chaudhuri \etal~\cite{CSWZ00} show that all graphs with $k\leq 5$
  terminals have cut-sparsifiers with quality $\beta=1$ and at most 
  $k+1$ vertices. (See also \cite[Table 1]{KRTV12}.)
\end{myproof}

The above two results are direct corollaries, but we can use
Theorem~\ref{thm:cut-flow} to get flow-sparsifiers with quality~1 from
results on cut-sparsifiers, even when the flow-cut gap is more than
$1$. E.g., for series-parallel graphs we know that the flow-cut gap is 
exactly~2~\cite{CJLV08,CSW13,LR10}, but we give in the next section
quality $1$ flow-sparsifiers by using cut-sparsifiers more directly.

\subsection{Series-Parallel Graphs and Graphs of Bounded Treewidth}

To begin, we give some definitions. An $s$-$t$ \emph{series-parallel
  graph} is defined recursively: it is either~(a)~a single edge
$\aset{s,t}$, or (b)~obtained by taking a parallel composition of two
smaller $s$-$t$ series-parallel graphs by identifying their $s$ and $t$
nodes, or (c)~obtained by a series composition of an $s$-$x$
series-parallel graph with an $x$-$t$ series-parallel graph 
by identifying their $x$ node. See
Figure~\ref{fig:sp}. The vertices $s, t$ are called the \emph{portals}
of $G$, and the rest of the vertices will be called the \emph{internal}
vertices.

\begin{figure}[h]
\centering
\includegraphics[scale=0.5]{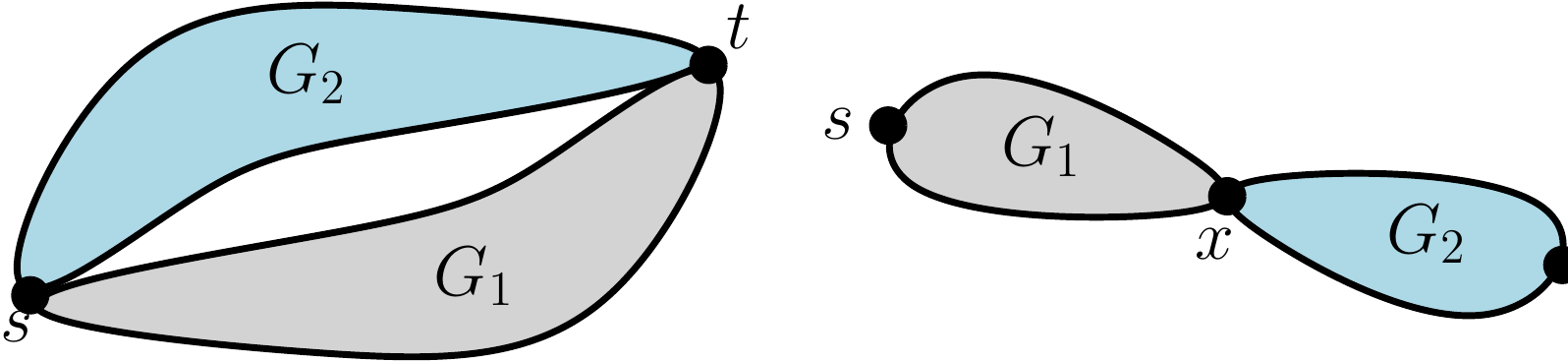}
\caption{Series and Parallel Compositions}
\label{fig:sp}
\end{figure}

\begin{theorem}[Series-Parallel Graphs]
  \label{thm:sp}
  Every $k$-terminal series-parallel network $G$ admits an exact (quality $1$)
  flow sparsifier with $O(k)$ vertices.
\end{theorem}

\begin{myproof}
  The way we build a series-parallel graph $G$ gives us a decomposition tree
  $\T$, where the leaves of $\T$ are edges in $G$, and each internal
  node prescribes either a series or a parallel combination of the
  graphs given by the two subtrees.  We can label each node in $\T$ by
  the two portals. We will assume w.l.o.g.\ that $\T$ is binary.

  Consider some decomposition tree $\T$ where the two portals for the
  root node are themselves terminals, and let the number of internal
  terminals in $\T$ be $k$ (giving us a total of $k+2$ terminals,
  including the portals). 
  We construct a sparsifier for $G$ by working on the decomposition 
  tree $\T$ recursively as follows, producing a sparsifier graph with
  at most $S(k)$ vertices, for $S(k)$ that will be determined later.
  Consider the two subtrees $\T_1, \T_2$. The
  easy case is when the number of internal terminals in $\T_1, \T_2$, 
  which we denote $k_1, k_2$, are both strictly less than $k$. Let $G_i$ be the
  graph defined by $\T_i$. In case $G_1, G_2$ are composed in parallel,
  recursively construct for them sparsifiers $G_1', G_2'$, and compose
  these two sparsifiers in parallel to get $G'$; 
  the Composition Lemma~\ref{lem:compose} implies this is a
  sparsifier for $G$. In case they are in series, the middle vertex may
  not be a terminal: so add it as a new terminal, recurse on $G_1, G_2$, 
  and again compose $G'_1, G'_2$. 
  In either case, the number of internal vertices in the new
  graph is $S(k) \leq S(k_1) + S(k_2) + 1$, where we account for adding
  the middle vertex as a terminal.

  Now suppose all the $k$ internal terminals of the root of $\T$ are
  also internal terminals of $\T_1$. In this case, find the node in $\T$
  furthest from the root such that the subtree $\T'$ rooted at this
  node still has $k$ internal terminals, but neither of its child
  subtrees $\T_1', \T_2'$ contains all the $k$ internal
  terminals.\footnote{The easy case above is in fact a special case of this,
    where $\T' = \T$.} 
  There must be such a node, because the leaves of
  $\T$ contain no internal terminals. Say the portals of $\T'$ are $s',
  t'$. And say the graphs given by $\T', \T_1', \T_2'$ are $G', G_1',
  G_2'$. The picture looks like one of the cases in
  Figure~\ref{fig:recursion}. 

  \ifstoc
  \begin{figure}[h]
    \centering
    \includegraphics[scale=0.4]{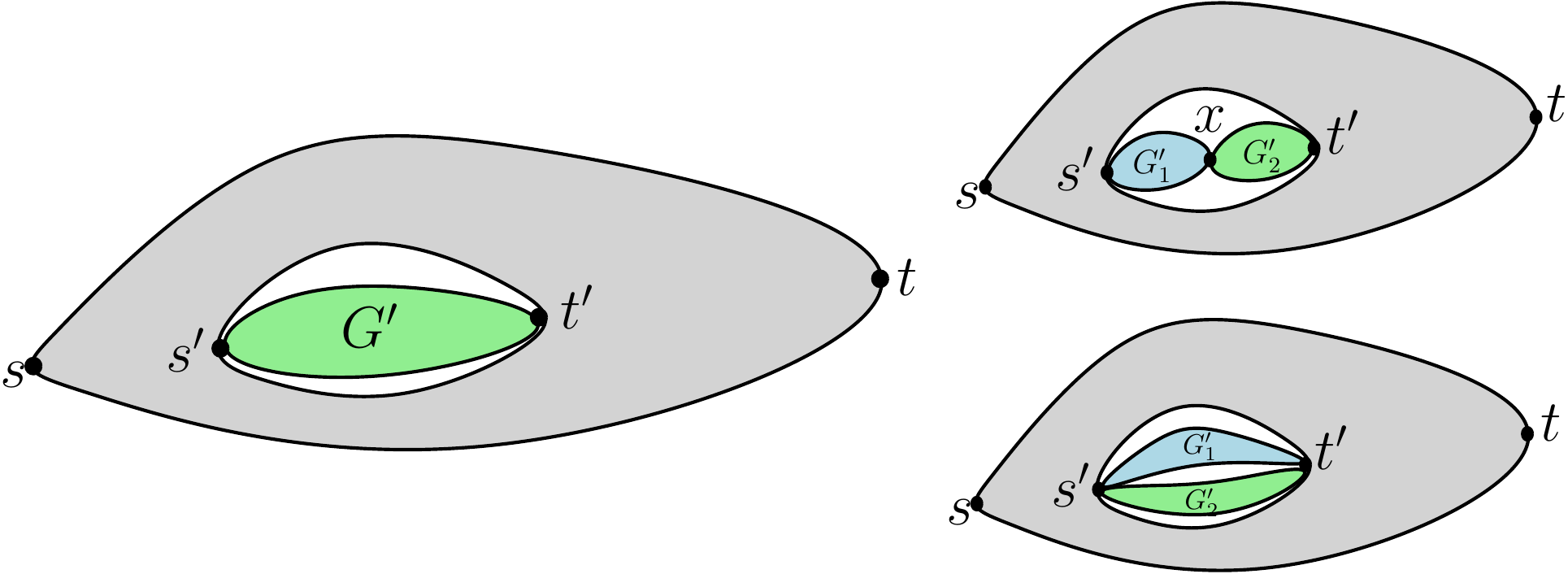}
    \caption{The subgraph $G'$ within $G$}
    \label{fig:recursion}
  \end{figure}
  \else
  \begin{figure}[h]
    \centering
    \includegraphics[scale=0.5]{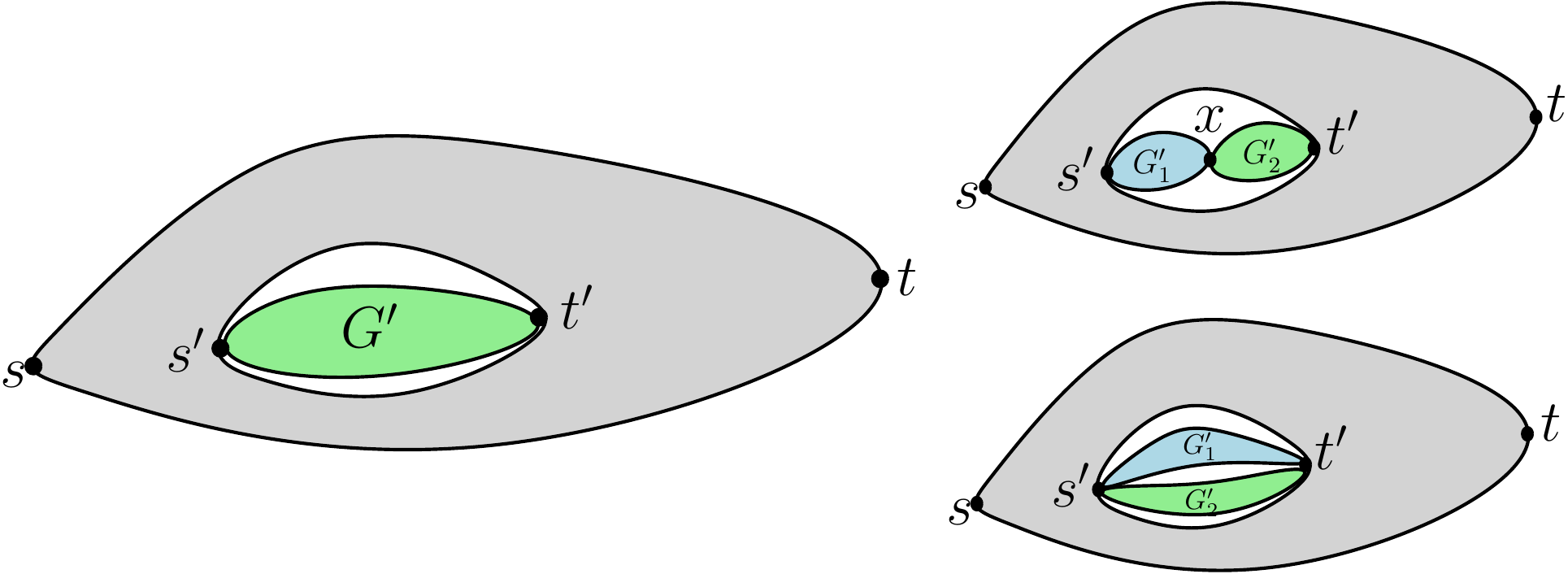}
    \caption{The subgraph $G'$ within $G$}
    \label{fig:recursion}
  \end{figure}
  \fi

  Add the two portals $s', t'$ of $G'$, and if it was a series combination
  then also the middle vertex $x$ between $G_1', G_2'$, as new terminals.
  Observe that $G$ is obtained by composing $G\setminus G'$ with $G'$
  (at the new terminals $s',t'$), hence we can apply the Composition
  Lemma to the sparsifiers for $G \setminus G'$ and for $G'$. We can use
  Corollary~\ref{cor:translate-K4} to find a flow sparsifier for $G
  \setminus G'$: it has only $4$ terminals $s,t,s',t'$. To find a flow
  sparsifier for $G'$, we recurse on $G_1', G_2'$ and then combine the
  resulting sparsifiers $H_1', H_2'$ by the Composition Lemma to get a
  sparsifier $H'$ for $G'$.  Overall, we obtain a flow sparsifier for
  $G$ with at most $S(k) \leq S(k_1) + S(k_2) + (c_4 - 4) + 3$ internal
  vertices, where the number of new vertices generated by
  Corollary~\ref{cor:translate-K4} is at most $c_4 - 4$, and we added in
  at most $3$ new terminals (namely $s', t'$ and possibly
  $x$).  

  In either case, we arrive at the recurrence $S(k) \leq S(k_1) + S(k_2)
  + c_4$, where $k_1 + k_2 \leq k$ and $k_1, k_2 \geq 1$. The base case
  is when there are at most $2$ internal terminals, in which case we can
  use Corollary~\ref{cor:translate-K4} again to get $S(1), S(2) \leq
  c_4$.  The recurrence solves to $S(k) \leq (2k - 1)\cdot c_4$.
  Adding the two portal terminals of $\T$ still remains $O(k)$, and
  proves the theorem.
\end{myproof}

\subsection{Extension to Treewidth-$w$ Graphs}
\label{sec:treewidth}

The general theorem about bounded treewidth graphs follows a similar argument
but with looser bounds.  The only fact about a treewidth-$w$ graph $G = (V,E)$
we use is the following.
\begin{theorem}[\cite{Reed92}]
  \label{thm:reed}
  If a graph $G = (V,E)$ has treewidth $w$, then for every subset $T \sse
  V$, there exists a subset $X \sse V$ of $w$ vertices such that each
  component of $G - X$ contains at most $\frac23|T \setminus X|$
  vertices of $T$.
\end{theorem}

\begin{theorem}
  \label{thm:twk}
  Suppose every $k$-terminal network admits a flow sparsifier of quality
  $q(k)$ and size $S(k)$. Then every $k$-terminal network $G$ with
  treewidth $w$ has a $q(6w)$-quality flow sparsifier with at most $k^4
  \cdot S(6w)$ vertices.
\end{theorem}

\begin{myproof}
  The proof is by induction. Consider a graph $G$: if it has at most
  $6w$ terminals, we just build a $q(6w)$-quality vertex sparsifier of
  size $S(6w)$.

  Else, let $T$ be the set of terminals in $G$, and use
  Theorem~\ref{thm:reed} to find a set $X$ such that each component of
  $G - X$ contains at most $\frac23|T \setminus X|$ terminals. Suppose
  the components have vertex sets $V_1, V_2, \ldots, V_l$; let $G_i :=
  G[V_i \cup X]$. Recurse on each $G_i$ with terminal set $(T \cap V_i)
  \cup X$ to find a flow sparsifier $G_i'$ of quality $q(6w)$.  Now use
  the Composition lemma to merge these sparsifiers $G_i'$ together and
  give the sparsifier $G'$ of the same quality.  Now use the Composition
  lemma to merge these sparsifiers $G_i'$ together and give the
  sparsifier $G'$ of the same quality.

  If the number of terminals in $G$ was $k_G$, the number of terminals
  in each $G_i$ is smaller by at least $\frac13 k_G - w = k_G/6$, and hence
  $k_{G_i} \leq 5/6\, k_G$. Hence the depth of the recursion is at most
  $h := \log_{6/5} (k/w) \leq \log_{6/5} k$, and the number of leaves is
  at most $2^h$. Each leaf gives us a sparsifier of size $S(6w)$, and
  combining these gives a sparsifier of size at most $S(6w) \cdot
  k^{\log_{6/5} 2} \leq S(6w) \cdot k^4$.
\end{myproof}
Using, e.g., results from Englert et al.~\cite{EGKRTT10} we can achieve
$q(k) = O\big( \frac{\log k}{\log\log k} \big)$ and $S(k) = k$, which
gives the results stated in Section~\ref{sec:intro}.

{ \small
\bibliographystyle{alphainit}
\bibliography{robi}
}

\end{document}